\documentclass[letterpaper]{article} 
\usepackage{aaai2026}  
\usepackage{times}  
\usepackage{helvet}  
\usepackage{courier}  
\usepackage[hyphens]{url}  
\usepackage{graphicx} 
\usepackage{natbib}  
\usepackage{caption} 
\usepackage{algorithm}
\usepackage{algorithmic}
\usepackage{amssymb}
\usepackage{amsmath}
\usepackage{xcolor}
\usepackage{tabularx} 
\usepackage{graphicx}
\usepackage{array}    
\usepackage{subfig}
\usepackage{enumerate}
\newtheorem{definition}{Definition}
\newtheorem{theorem}{Theorem}

\newenvironment{proof}{{\noindent\it Proof}\quad}{\hfill $\square$\par}
\usepackage{newfloat}
\usepackage{listings}
\DeclareCaptionStyle{ruled}{labelfont=normalfont,labelsep=colon,strut=off} 
\floatstyle{ruled}
\newfloat{listing}{tb}{lst}{}
\floatname{listing}{Listing}
\title{A Historical Interaction-Enhanced Shapley Policy Gradient Algorithm for Multi-Agent Credit Assignment}
\author{
    Ao Ding\textsuperscript{\rm 1},
    Licheng Sun\textsuperscript{\rm 1},
    Yongjie Hou\textsuperscript{\rm 1},
    Huaqing Zhang\textsuperscript{\rm 1},
    Hongbin Ma\textsuperscript{\rm 1}
}
\affiliations{
    \textsuperscript{\rm 1}Beijing Institute of Technology\\
}

\usepackage{bibentry}

\begin{document}

\maketitle

\begin{abstract}
Multi-agent reinforcement learning (MARL) has demonstrated remarkable performance in multi-agent collaboration problems and has become a prominent topic in artificial intelligence research in recent years. However, traditional credit assignment schemes in MARL cannot reliably capture individual contributions in strongly coupled tasks while maintaining training stability, which leads to limited generalization capabilities and hinders algorithm performance. To address these challenges, we propose a Historical Interaction-Enhanced Shapley Policy Gradient Algorithm (HIS) for Multi-Agent Credit Assignment, which employs a hybrid credit assignment mechanism to balance base rewards with individual contribution incentives. By utilizing historical interaction data to calculate the Shapley value in a sample-efficient manner, HIS enhances the agent's ability to perceive its own contribution, while retaining the global reward to maintain training stability. Additionally, we provide theoretical guarantees for the hybrid credit assignment mechanism, ensuring that the assignment results it generates are both efficient and stable. We evaluate the proposed algorithm in three widely used continuous-action benchmark environments: Multi-Agent Particle Environment, Multi-Agent MuJoCo, and Bi-DexHands. Experimental results demonstrate that HIS outperforms state-of-the-art methods, particularly excelling in strongly coupled, complex collaborative tasks.
\end{abstract}


\section{Introduction}
Multi-agent reinforcement learning (MARL) has found important applications across various domains, enabling the solution of challenging collaborative tasks in practical scenarios such as autonomous driving \cite{zhou2024multiagent}, drone formation control \cite{hu2021distributed}, network management \cite{naderializadeh2021resource}, and power and energy distribution \cite{hua2023energy}.
Traditional independent learning frameworks face issues like incomplete information sharing, lack of coordination, and suboptimal performance in collaborative tasks. 
The centralized training and decentralized execution (CTDE) paradigm was proposed to address these challenges and has been adopted by many MARL algorithms \cite{lowe2017multi, zhong2024heterogeneous}.
In this framework, agents aggregate interaction information during training but make decisions based solely on local observations during execution, preserving decentralization.
In MARL, the credit assignment problem arises from the difficulty of decomposing the global reward, limiting performance by hindering agents' ability to accurately estimate their contributions.

The shared reward scheme distributes the global reward equally among all agents, thereby assigning the same optimization reward to each agent. Algorithms such as HAPPO and HATRPO \cite{kuba2021trust}, which are based on the HAML template, ensure monotonic improvement of the joint objective and convergence to a Nash equilibrium. HASAC \cite{liu2023maximum}, utilizing the MEHAML framework, further enhances exploration capabilities to guarantee convergence to a quantum response equilibrium.
However, this solution has difficulties in credit allocation, and the agent's inability to measure its own contribution leads to vague responsibilities, which may lead to the system convergence to a suboptimal solution.

The local reward scheme decomposes the global reward based on the contribution of each agent, aiming to clarify and optimize the contribution of each agent's individual policy. This approach can be further classified into implicit and explicit methods \cite{li2021shapley}:
(1) Implicit method: The global action value function relies on the decomposition function to implicitly assign credits. During the training process, the centralized critic and the local agent are trained as a whole \cite{sunehag2017value}.
(2) Explicit method: In this approach, the value function for each agent is derived by explicitly decomposing the joint value function, which is then used to guide policy updates \cite{foerster2018counterfactual}. 
The local reward scheme can be framed within the context of convex games, a typical cooperative game model where the objective is to divide the coalition and reach an agreement among agents within the same coalition \cite{chalkiadakis2011computational}. 
In a convex game, a stable profit distribution scheme in the Core can ensure that the coalition structure remains stable and no agent has the motivation to leave to seek greater personal benefits \cite{wang2022shaq}.
When all agents are grouped into a single large coalition, the reward distribution scheme effectively serves as a mechanism for credit assignment. However, in highly coupled collaborative scenarios, accurately decomposing the contribution of individual actions to the global reward remains a challenge. 
In such cases, simple credit assignment methods may lead to inaccurate estimation of policy contributions for some agents due to calculation errors, resulting in unstable training dynamics.

To address the challenges outlined above, this paper introduces the Historical Interaction-Enhanced Shapley Policy Gradient Algorithm (HIS) for Multi-Agent Credit Assignment. 
The hybrid assignment mechanism it introduces can not only explainably calculate the contribution of each agent's strategy to the group and provide excess rewards for high-contribution agents, but also provide basic rewards for all agents compared to the pure trust assignment scheme, ensuring the robustness and stability of training. 

Our main contributions can be summarized as follows:
\begin{itemize}
\item We analyze the limitations of current MARL credit assignment schemes and propose a hybrid assignment mechanism to solve the credit assignment problem.
\item We provide theoretical guarantees for this hybrid assignment mechanism, showing that it ensures stable training and differentiates agent contributions.
\item We propose the HIS algorithm, which uses historical interactions to calculate the Shapley Q-value of each agent's policy, encouraging agents to recognize their contribution to the overall task. Experiments show that the sophisticated design of HIS enables it to perform well in a variety of multi-agent tasks, especially in strongly coupled robotics tasks.
\end{itemize}

\section{Related Work}
In recent years, MARL has demonstrated strong performance in cooperative game tasks across various domains \cite{huang2024multi, jiang2025qllm}, prompting researchers to explore its application in more complex cooperative game scenarios \cite{zhong2024heterogeneous}. However, the multi-agent credit allocation problem has limited MARL training performance. Ensuring that each agent receives a reward corresponding to its individual contribution, while maintaining stable training, has become a critical challenge in the field.

In the shared reward scheme, all agents optimize a common reward signal to maximize the global objective \cite{wang2024shapley}.
HAPPO and HATRPO \cite{kuba2021trust} ensure monotonic improvement of the joint objective and convergence to a Nash equilibrium by employing a sequential update scheme within the HAML framework. 
HASAC \cite{liu2023maximum} enhances exploration by incorporating maximum-entropy MARL principles and introduces the MEHAML template, which preserves the monotonic improvement and convergence properties of the algorithm.
Nevertheless, the shared reward scheme still suffers from significant credit assignment challenges, as it fails to attribute the global reward to individual agents' contributions. Consequently, it remains difficult to assess and optimize each agent's specific impact on the collective outcome.

Under the local reward scheme, agents decompose the global reward using either implicit or explicit methods and optimize their strategies based on their individual contributions \cite{wang2020shapley}.
The most common implicit decomposition method in MARL is based on value decomposition. For example, VDN \cite{sunehag2017value} decomposes the global value function as a sum of individual agent value functions, enabling credit assignment.
QMIX \cite{rashid2020monotonic} introduces monotonic functions under the Individual-Global-Max (IGM) principle to enhance the expressive power of value decomposition.
FACMAC \cite{peng2021facmac} introduced the idea of value decomposition of QMIX into policy-based MARL without having to follow the IGM principle, and achieved outstanding performance in multi-agent collaboration tasks.
However, implicit methods rely heavily on the quality of the decomposition function and lack interpretability in quantifying each agent's contribution.
Explicit methods, in contrast, compute attributable contributions for each agent directly.
COMA \cite{foerster2018counterfactual} addresses the credit assignment problem using a counterfactual advantage baseline, which compares an agent's expected return to a counterfactual scenario to assess its contribution to the overall task.
However, this method ignores the correlation between agents and performs poorly in complex scenarios.
SQDDPG \cite{wang2020shapley} uses Shapley value to solve the credit assignment problem of MARL. It proposes a network to estimate marginal contribution, which is then used to approximate the calculation of Shapley value and is used to guide the learning of local agents.
Nevertheless, explicit methods still rely on accurately calculated contribution data to update strategies. In scenarios where contributions cannot be accurately estimated, some agents may not receive sufficient rewards for a long time, resulting in impaired learning effects and affecting the stability of training.
\section{Preliminaries}
\subsection{Cooperative Multi-Agent Reinforcement Learning}
A cooperative Markov game \cite{littman1994markov,wu2024learning} can be represented by a tuple $\langle\mathcal{N},\mathcal{S},\mathcal{A},r,P,\gamma,d\rangle$.
Where $\mathcal{S}$ is the state space. $\mathcal{A}$ is the joint action space. $ P ( s ^ { \prime } | s , a )$ is the state transition probability from state $s$ to state $ s ^ { \prime }$ given the joint action $ a = ( a ^ { 1 } , \ldots , a ^ { n } )$. 
$r(s,a)$ is the joint reward function. $\mathcal{N}=\{1,\ldots,n\}$ is the set of n agents.
$\gamma$ is the discount factor. $d$ is the initial state distribution, where $ d \in \mathcal{P} ( X )$. $\mathcal{P} ( X )$ is the set of state distributions of set $X$, which is used in this paper to identify the power set of set $X$.
$Sym(n)$ is used to represent the set of permutations of the integer set $ \left\{ 1 , \ldots , n \right\}$. At each time step $t$, each agent $i \in \mathcal{N}$ selects an individual action $a_t^i$ based on the current state $s_t$ and its policy $\pi^i$. 
The joint action at time $t$ is therefore given by $a_t = (a_t^1, \ldots, a_t^n)$. The joint policy in this process is defined as $ \pi ( \cdot | s _ { t } ) = \prod _ { i = 1 } ^ { n } \pi ^ { i } ( \cdot ^ { i } | s _ { t } )$.
The marginal state distribution $ \rho _ { \pi } ^ { t }$ at time $t$ is determined by the initial state distribution $d$, the joint policy $\pi$ and the transition probability $P$ \cite{liu2023maximum}.
The unnormalized marginal state distribution is defined as $ \rho _ { \pi } \triangleq \sum _ { t = 1 } ^ { T } \rho _ { \pi } ^ { t }$. The goal of all agents is to maximize the expected total reward, defined as:
\begin{equation} \label{eq-J}
    J _ { s t d } ( \pi ) = E _ { s _ { 1 : T}  \sim \rho _{\pi}^{1 : T} ,a _ { 1 : T } \sim \pi} \left[ \sum _ { t = 1 } ^ { T } r ( s _ { t } , a _ { t } ) \right].
\end{equation}

\subsection{Convex Game}
In cooperative game theory, a convex game (CG) is a representative model within the class of transferable utility games \cite{chalkiadakis2011computational}.
A CG is defined by the tuple $(\mathcal{N}, v)$, where $\mathcal{N} =  \left\{ 1 , \ldots , n \right\}$ represents the grand coalition of $n$ agents, and
$v : 2^{n} \rightarrow \mathbb{R}$ is the value function, which assigns a value $v(\mathcal{C})$ to any coalition $\mathcal{C} \subseteq \mathcal{N}$.
The value function of CG has the following properties \cite{lee2024solutions}: (1) Superadditivity: $ v ( \mathcal{C} \cup \mathcal{D} ) \geq v ( \mathcal{C} ) + v ( \mathcal{D} ),\forall \mathcal{C} ,\mathcal{ D} \subset \mathcal{N} , \mathcal{C} \cap \mathcal{D} = \mathcal{\emptyset}$; (2) Independence: The coalitions are independent.
A solution to CG is a tuple $ ( \mathcal{C S} , x )$ , where $\mathcal{CS} = \{\mathcal{C}^1, \ldots, \mathcal{C}^k\}$ is a possibility partition of the grand coalition $\mathcal{N}$.
$x=(x^{1},\ldots,x^{n}) \in \mathbb{R} ^ { n }$ is the corresponding profit distribution vector, which satisfies two properties \cite{chen2024approaching}: (1)  $ x ( \mathcal{C} ^ { j } ) \leq v ( \mathcal{C} ^ { j } )$, where $ x ( \mathcal{C} ^ { j } ) = \sum _ { i \in \mathcal{C} ^ { j } } x ^ { i }$, for any $j \in \left\{ 1 , \ldots , k \right\}$. (2) $x ^ { i } \geq 0$ for each $i \in \mathcal{N}$.
\begin{definition}
    A transferable utility game $( \mathcal{N} , v )$ is a convex game if for $ \forall \mathcal{C} , \mathcal{D} \subseteq N$ the following condition holds:
    \begin{equation} \label{definition-1}
        v ( \mathcal{C} \cup \mathcal{D} ) + v ( C \cap D  ) \geq v ( C  ) + v ( D ).
    \end{equation}
\end{definition}
\begin{definition} \label{definition-2}
    If for any $\cal {C} \in \cal{CS}$, in the outcome $(\mathcal{CS}, x)$, $x(\mathcal{C}) = v(\mathcal{C})$, then this outcome is efficient and maximizes social welfare.
\end{definition}

\begin{definition} \label{definition-3}
    Given a convex game $(\mathcal{N},v)$, the Core is the set of all stable solutions. For each $\mathcal{C} \in \mathcal{N}$, there is $ x ( C ) \geq v ( C )$, where $ x ( \mathcal{C} ) = \sum _ { i \in \mathcal{C} } x ^ { i }$.
\end{definition}

\begin{definition}
    The Core of a convex game must be non-empty.
\end{definition}

\subsection{Shapley Value}
The Shapley value is a widely used solution concept for distributing revenue in the grand coalition \cite{shapley1953value}, and it has been applied in various areas of machine learning to explain feature contributions \cite{lundberg2017unified, jethani2021fastshap,li2024shapley}.
Given a convex game $(\mathcal{N}, v)$, the Shapley value of agent $i$ is its assigned payoff $x_i$, denoted $\Phi_i$, is defined as:
\begin{equation} \label{eq-sharpvalue}
    \Phi_i=\sum_{\mathcal{C}\subseteq\mathcal{N}\setminus\{i\}}\frac{|\mathcal{C}|!(|\mathcal{N}|-|\mathcal{C}|-1)!}{|\mathcal{N}|!}\cdot\delta_i(\mathcal{C}).
\end{equation}
where $\delta_i(\mathcal{C})= v ( \mathcal{C} \cup \left\{ i \right\} ) - v ( \mathcal{C} )$ is the marginal contribution of agent $i$ to coalition $\mathcal{C}$.
The Shapley value \cite{chen2024stas} is computed as the weighted average of an agent's marginal contributions across all possible sub-coalitions. It satisfies two key properties: (1) Dummy player: If agent $i$ does not make any contribution, then the allocated contribution $x_i = 0$. (2) Symmetry: If the contributions made by agent $i$ and agent $j$ are the same, then the allocated contributions obtained by the two agents are also the same, that is, $x_i=x_j$. (3) Efficiency: The total allocated payoff equals the value of the grand coalition, that is, $x(\mathcal{N}) = v(\mathcal{N})$.
\begin{theorem}
    Given a convex game $G=(\mathcal{N},v)$, the Shapley value must belong to the Core, and the assignment scheme is stable.
\end{theorem}

\begin{figure*}[t] 
    \centering
    \includegraphics[width=0.65\linewidth]{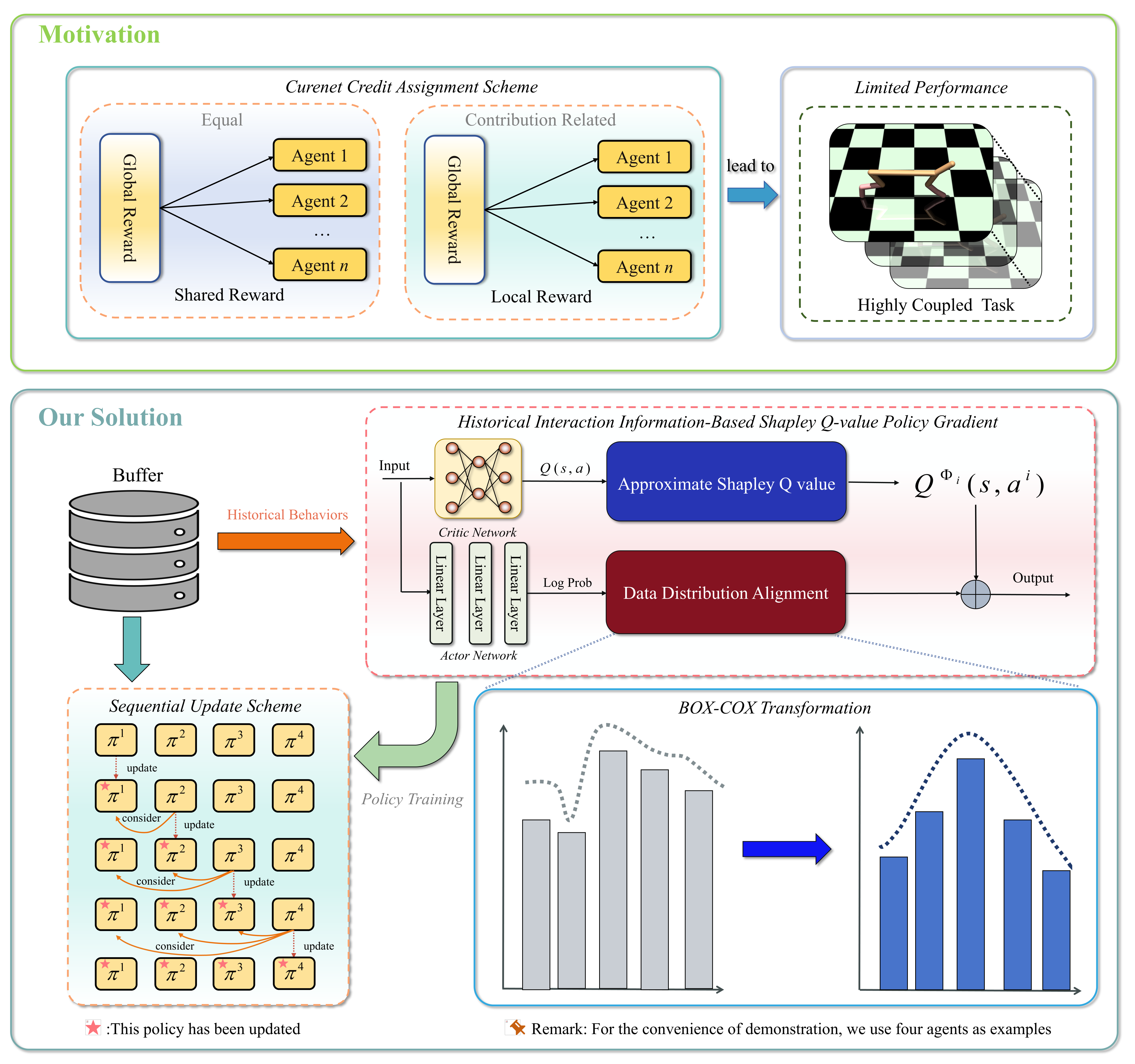}%
    \caption{Schematic diagram of the motivation and overall process of our proposed HIS method.}
    \label{fig:framework}
\end{figure*}

\subsection{BOX-COX Transformation}
The Box-Cox transformation is a generalized power transformation method \cite{liu2024box} commonly used in statistical modeling to improve the normality and symmetry of data while preserving the relative order of observations. In this paper, we use $\mathcal{BC}(\cdot)$ to denote the Box-Cox transformation. The standard form of Box-Cox transformation is:
\begin{equation} \label{eq_BC1}
    \begin{aligned}  
    {\cal BC}\left( x_d \right)=
    \begin{cases}
    \frac{x _ { d } ^ { \lambda } - 1}{\lambda} \hspace{5mm} & {\lambda \ne 0} \\
    \log x _ { d } & \lambda  = 0
    \end{cases}
    \end{aligned}
\end{equation}
where $x_d$ denotes the data to be transformed, and $\lambda$ is a parameter that controls the degree and type of transformation.
When the dataset $x_d$ contains negative values, the Box-Cox transformation \cite{chen2025investigations} is adjusted as follows:
\begin{equation} \label{eq_BC2}
    \begin{aligned}  
    {\cal BC}\left( x_d \right)=
    \begin{cases}
    \frac{\left(x_d-x_{min}+\varsigma \right)^\lambda-1} {\lambda} \hspace{5mm} & {\lambda \ne 0} \\
    \log \left(x_d - x_{min} + \varsigma \right) & \lambda = 0
    \end{cases}
    \end{aligned}
\end{equation}
where $x_{\min}$ is the minimum value in $x_d$, and $\varsigma$ is a small positive constant introduced to ensure that $x_d - x_{\min} + \varsigma > 0$.
In practice, the parameter $\lambda$ is typically estimated from the data using Bayesian or maximum likelihood methods. 
\section{Method}
In this section, we construct a Historical Interaction-Enhanced Shapley Policy Gradient Algorithm (HIS), which leverages historical interaction data to compute Shapley values and distributes contributions to agents through a hybrid credit assignment mechanism, thereby enhancing the agents' understanding of their personal policies.
We will explain the algorithm in three parts. Subsection 4.1 introduces the hybrid credit assignment mechanism and provides a theoretical justification for its design.
Subsection 4.2 presents a novel modeling approach for applying the Shapley value to MARL, enabling fast estimation of the Shapley value based on historical interactions, which facilitates explicit credit assignment to multiple agents. 
Subsection 4.3 details the implementation of the Historical Interaction-Enhanced Shapley Policy Gradient. The motivation for proposing HIS and its overall process are intuitively shown in Figure 1.

\subsection{Hybrid Credit Assignment Mechanism} \label{sec-hybrid}
The concept of the hybrid credit assignment mechanism is inspired by the base salary plus performance-based bonus system \cite{zulkarnain2024optimizing} commonly used in modern enterprises.
This hybrid credit distribution model cleverly balances fairness and incentives: it ensures stable learning for all participants by providing basic rewards, while giving differentiated rewards based on the strategic contribution of each participant.

The hybrid credit assignment mechanism splits the value function $v(\cal{C})$ into two equal parts. The first half serves as a global reward, which is evenly distributed among all agents to ensure baseline stability. The second half is used as the basis for computing the Shapley value, which reflects each agent's strategic contribution and serves as an additional incentive.
We will prove that given a convex game $G=(\mathcal{N},v)$, this hybrid credit assignment yields an outcome $(\mathcal{N},x)$ that is both efficient and stable.
Here, $x=( x ^ { 1 } , \ldots , x ^ { n } )$ denotes the payoff vector allocated to the agents, and $\mathcal{N}$ represents the grand coalition of all agents. Specifically, each agent will obtain the distributed payoff $x^{i}=(\frac{v(\mathcal{C}) }{2*|\mathcal{N}|}+\Phi_{i})$ according to its own policy.
$\frac{v(\mathcal{C}) }{2*|\mathcal{N}|}$ is the equal share of the first half of the global reward enjoyed by the agent, and $\Phi_{i}$ is the Shapley value of agent $i$ calculated according to the second half of the value function.

\begin{theorem} \label{theorem-2}
    Given a convex game $G=(\mathcal{N},v)$, the hybrid assignment outcome $(\mathcal{N},x)$ is efficient, and the payoff vector $x$ is an efficient reward reallocation.
\end{theorem}
\begin{theorem} \label{theorem-3}
    Given a convex game $G=(\mathcal{N},v)$, the hybrid assignment outcome $(\mathcal{N},x)$ is in the Core, and the payoff vector $x$ is a stable reward reallocation.
\end{theorem}
The proofs of Theorem \ref{theorem-2} and Theorem \ref{theorem-3} can be found in Appendix B.1 and Appendix B.2, respectively.
According to Theorem \ref{theorem-2} and Theorem \ref{theorem-3}, given a convex game, the hybrid assignment outcome $(\mathcal{N},x)$ is efficient and stable.
Specifically, no agent in the grand coalition has an incentive to leave and form a separate coalition, as doing so would not yield a higher payoff.

\subsection{Historical Interaction Information-Based Shapley Q-value Policy Gradient} \label{sec-shapley}
The Shapley value provides a fair assessment of an agent's contribution under the current policy by summing the marginal contributions of individual agents. However, its high computational complexity poses significant challenges for its application in multi-agent settings.
To address this challenge, we draw inspiration from SQDDPG \cite{wang2020shapley} and propose a novel Shapley Q-value modeling method that employs the Approximate Marginal Contribution (AMC) technique, enabling efficient and stable resolution of the credit assignment problem in MARL.

AMC directly computes the marginal contribution $\delta(\mathcal{C})$ of each coalition while maintaining stability.
Specifically, it first models the global reward $v(\mathcal{C})$ of the coalition $\cal{C}$ as the value function $Q(s_t,a_t)$ and defines the distribution of an agent randomly joining an existing coalition $\mathcal{C}$ as $Pr(\mathcal{C}|{\mathcal{N}}\backslash\{i\})= | \mathcal{C} | ! ( | \mathcal{N} | - | \mathcal{C} | - 1 ) ! / | \mathcal{N} | !$.
Then, according to this definition, the marginal distribution contribution can be approximated:
\begin{equation}
    \hat{\Phi}_{i}(s,{ a}_{\mathcal{C}\cup\{i\}}):\mathcal{S}\times{\cal A}_{\mathcal{C}\cup\{i\}}\mapsto\mathbb{R},
\end{equation}

In this way, AMC accurately maps the action space of the state space of the coalition $\mathcal{C}$ and agent $i$ to a value, maintaining the property of marginal contribution. After the marginal contribution is approximated, Monte Carlo sampling $M$ times is used to estimate the Shapley Q-value of agent $i$.
The Shapley Q-value can be approximated by the following equation:
\begin{equation}
    \Phi_i(\mathcal{C}) = Q^{\pi_{\mathcal{C} \cup \{i\}}}(s, a_{\mathcal{C} \cup \{i\}}) - Q^{\pi_{\mathcal{C}}}(s, a_{\mathcal{C}}),
\end{equation}
\begin{equation}
    Q^{\Phi_{i}}(s,a^{i})\approx\frac{1}{M}\sum_{k=1}^{M}\hat{\Phi}_{i}(s,{a}_{C_{k}\cup\left\{i\right\}}),\;\forall C_{k}\sim P r(C|\mathcal{N}\backslash\{i\}).
\end{equation}

SQDDPG replaces all Q-value related parts in the DPG algorithm with Shapley Q-value. This modeling method seriously affects the training speed and stability of the algorithm.
To address this issue, we propose a stable and efficient Shapley Q-value modeling method that requires estimating the Shapley value only once during the algorithm's execution, effectively solving the credit assignment problem in multi-agent systems. 
Specifically, we introduce a historical interaction information-based Shapley Q-value policy gradient algorithm, which aims to maximize the reward $ Q ^ { \Phi _ { i } } ( s , a ^ { i } )$ obtained by the current agent. 
The historical interaction information-based Shapley Q-value policy gradient algorithm is defined as:
\begin{equation} \label{J_Phi} 
    \! {\nabla_\theta}J_{i}\left(\theta \right)=\! \mathop \mathbb{E} \limits_{ \left({s}_t, {a}_t\right) \sim {\cal{D}}} \! \left[ \! \nabla_\theta  {\cal BC} \! \Big( \! \log {\pi _\theta } \! \left({a}^i_t|s_t \right) \! \Big)   Q^{\Phi_{i}}\left( s_t, {a}_t^i \right) \right].
\end{equation}
where $\cal{D}$ represents replay buffer, $\log {\pi _\theta } \! \left({a}^i_t|s_t \right)$ is the log probability of historical actions under the current policy.
This algorithm effectively leverages historical interactions to improve both the efficiency and stability of the credit assignment process.
The derivation of the policy gradient algorithm can be found in Appendix B.3.

The proposed Shapley Q-value policy gradient method, based on historical interaction information, utilizes a sample-efficient approach to evaluate the contributions of past strategies. This enables the agent to learn a policy that maximizes its own contribution. 
We argue that this approach more effectively assigns contributions to individual agents, thereby more clearly distinguishing individual policy performance based on each agent's historical behavior.

\subsection{Historical Interaction-Enhanced Shapley Policy Gradient} \label{sec-HIS}
\begin{figure*}[t] 
    \centering
    \subfloat[{Reference (continuous)}]{%
        \includegraphics[width=0.33\linewidth]{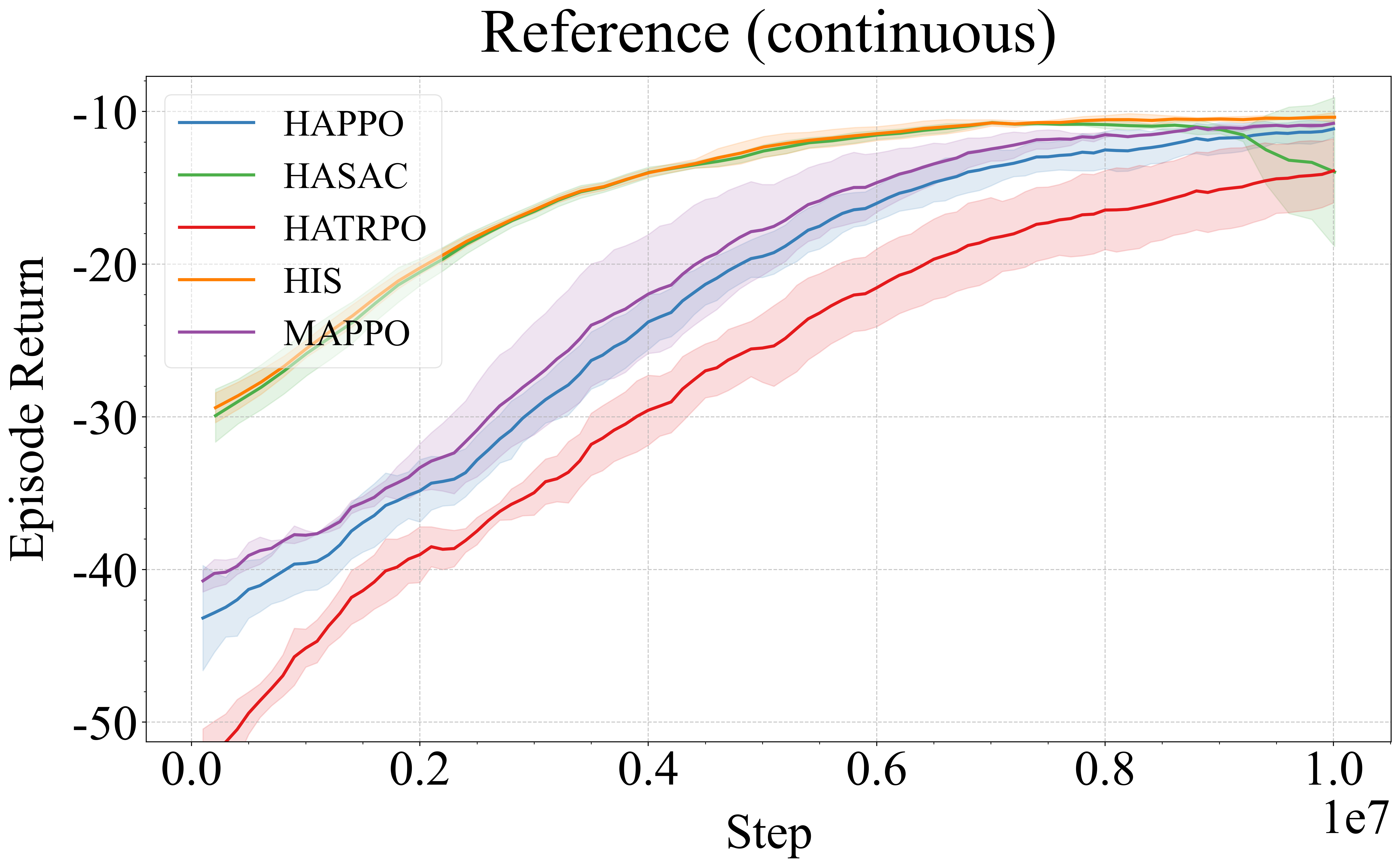}%
    }%
    \hfill%
    \subfloat[{Speaker Listener (continuous)}]{%
        \includegraphics[width=0.33\linewidth]{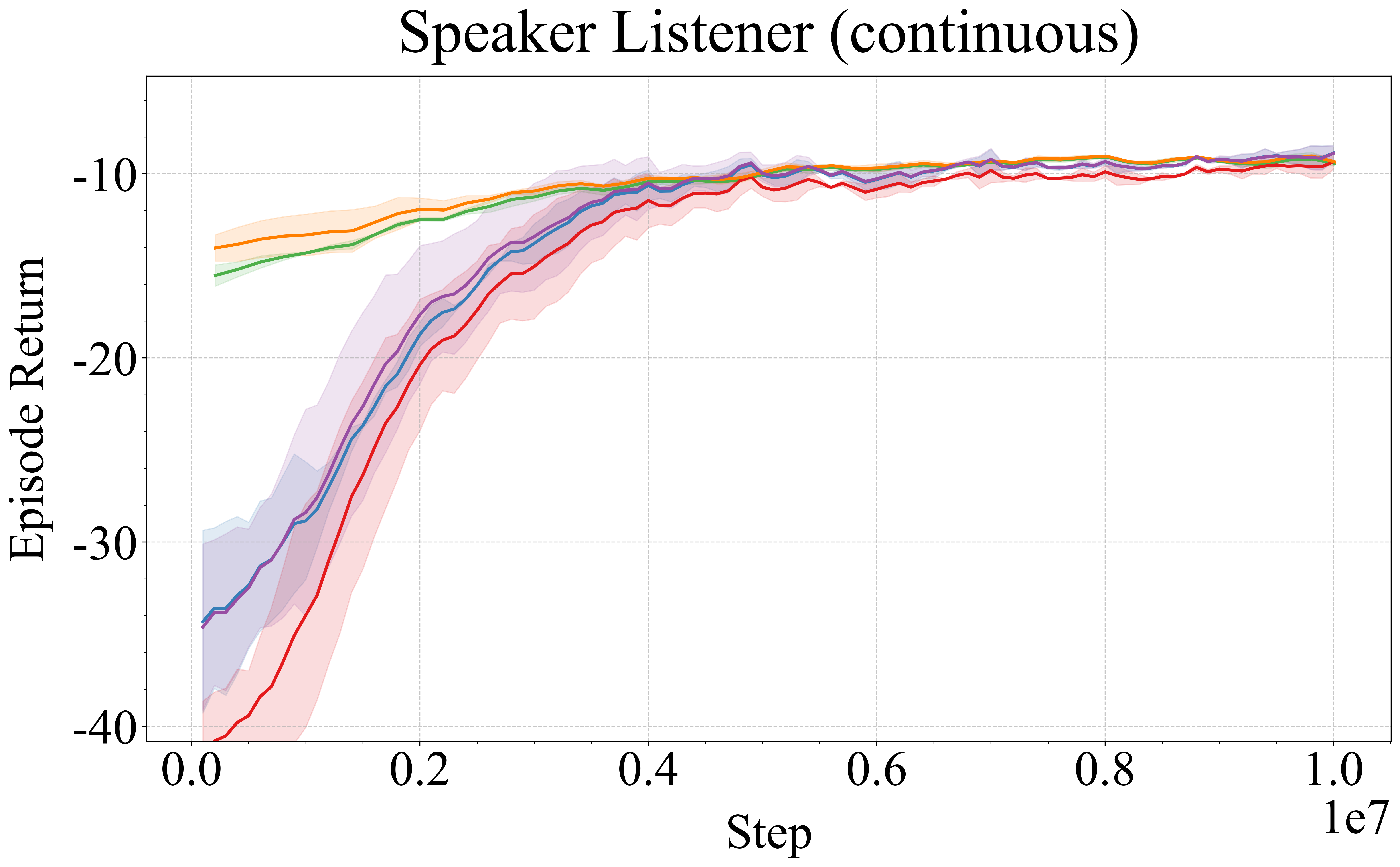}%
    }%
    \hfill%
    \subfloat[{Spread (continuous)}]{%
        \includegraphics[width=0.33\linewidth]{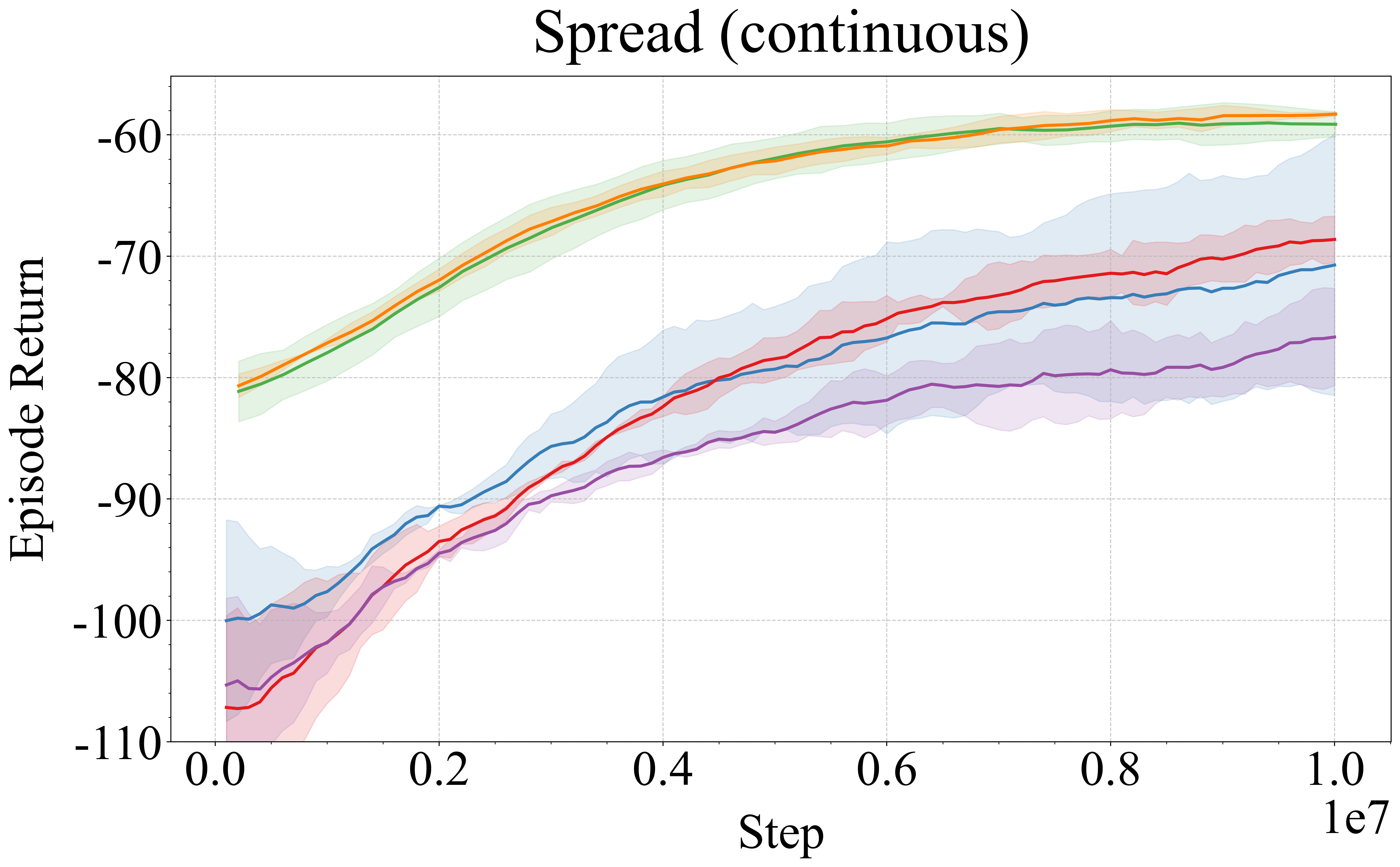}%
    }\\[1ex]
    \subfloat[{Catch Abreast}]{%
        \includegraphics[width=0.33\linewidth]{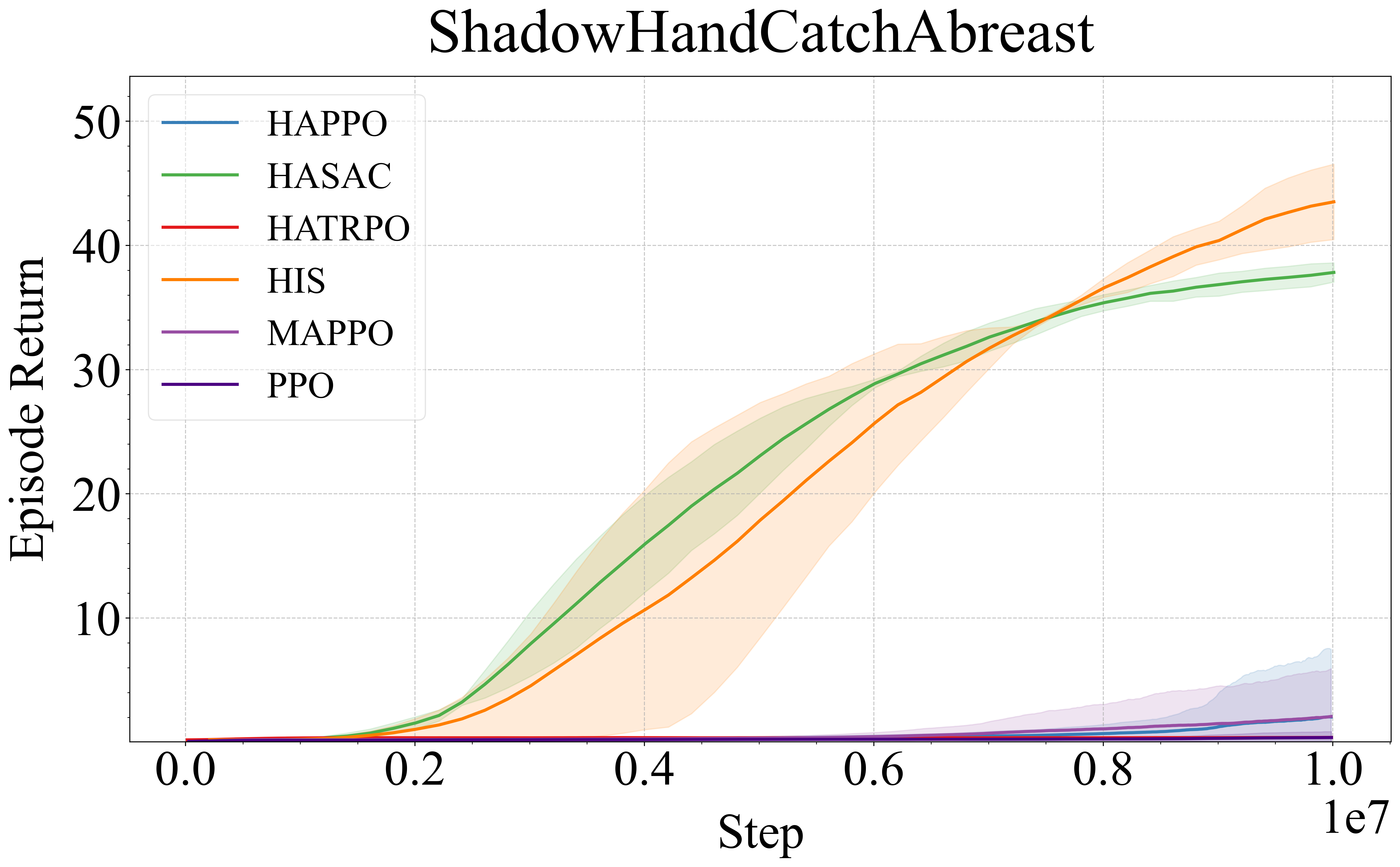}%
    }
    \hfill%
    \subfloat[{Catch Over2Underarm}]{%
        \includegraphics[width=0.33\linewidth]{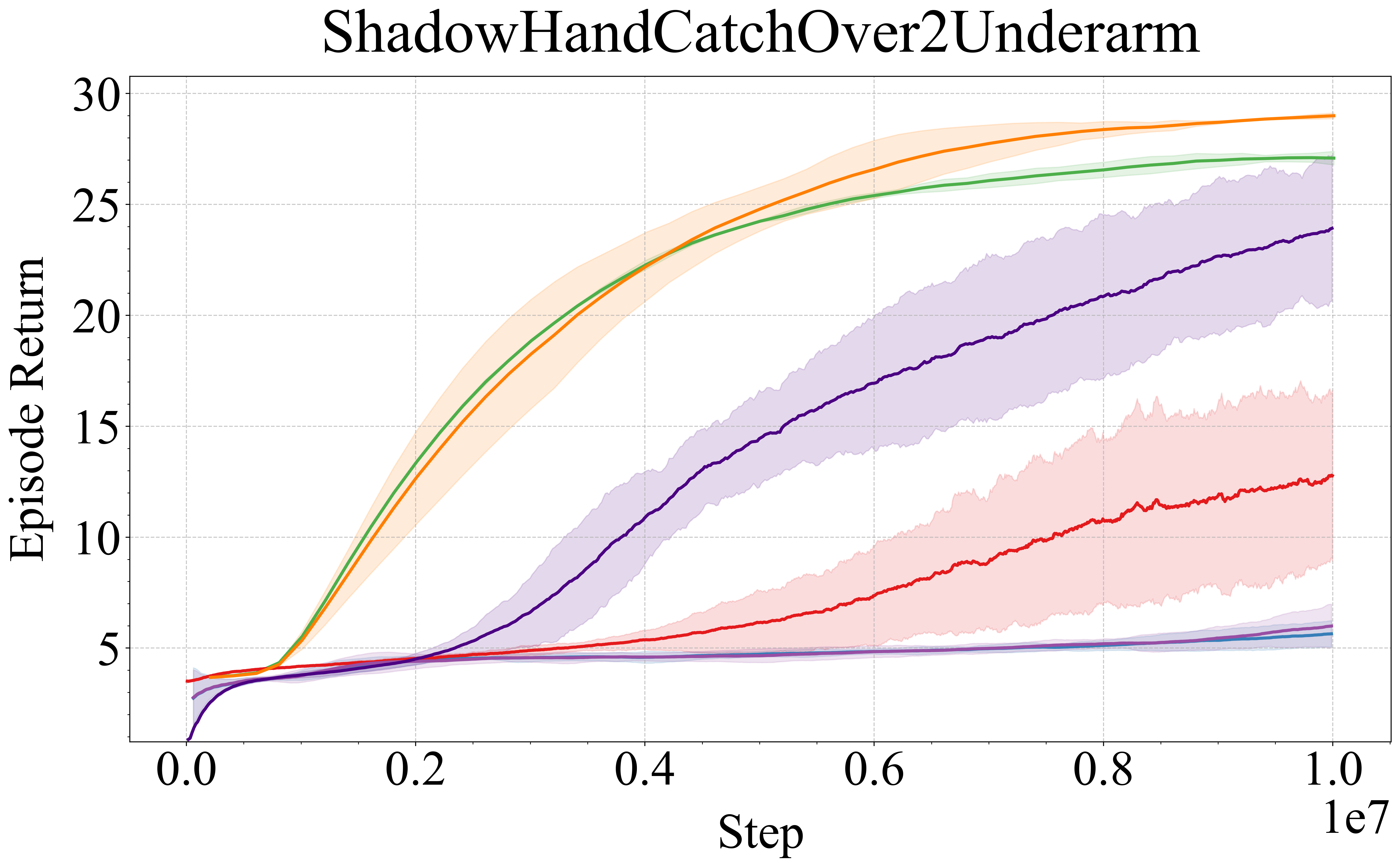}%
    }
    \hfill%
    \subfloat[{Hand Pen}]{%
        \includegraphics[width=0.33\linewidth]{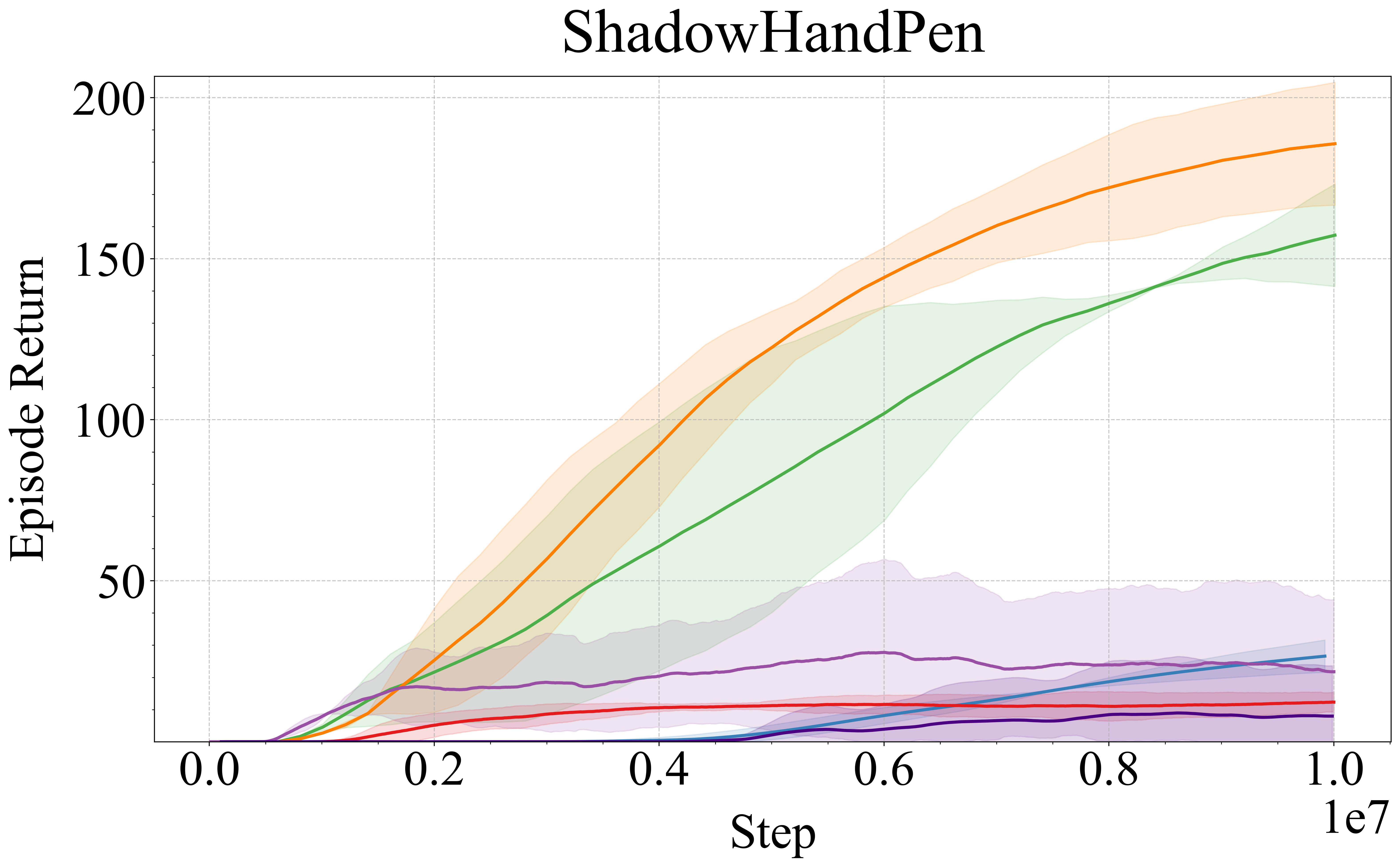}%
    }\\[1ex]
    \subfloat[{2x4-Agent Ant}]{%
        \includegraphics[width=0.33\linewidth]{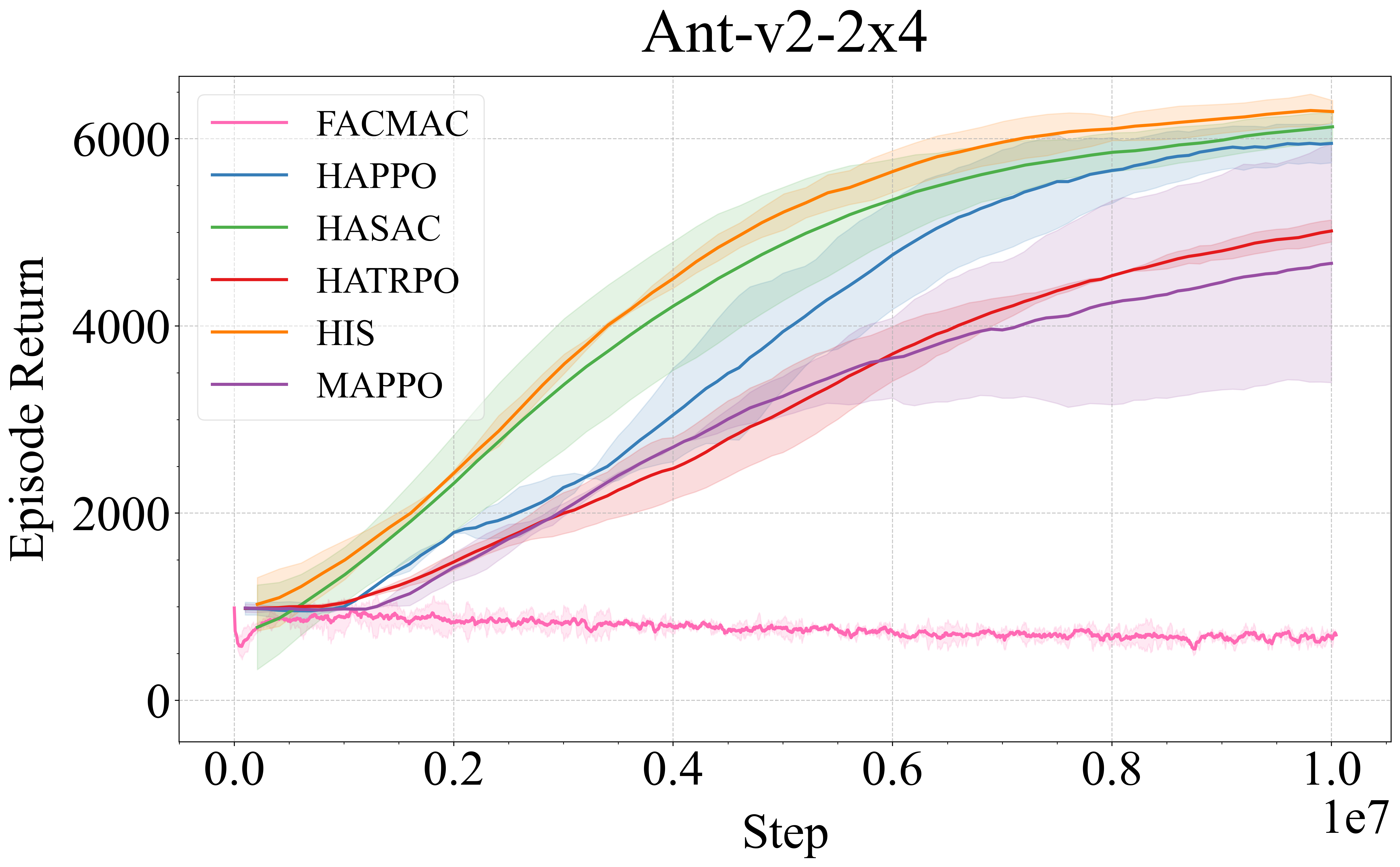}%
    }%
    \hfill%
    \subfloat[{2x3-Agent Walker2d}]{%
        \includegraphics[width=0.33\linewidth]{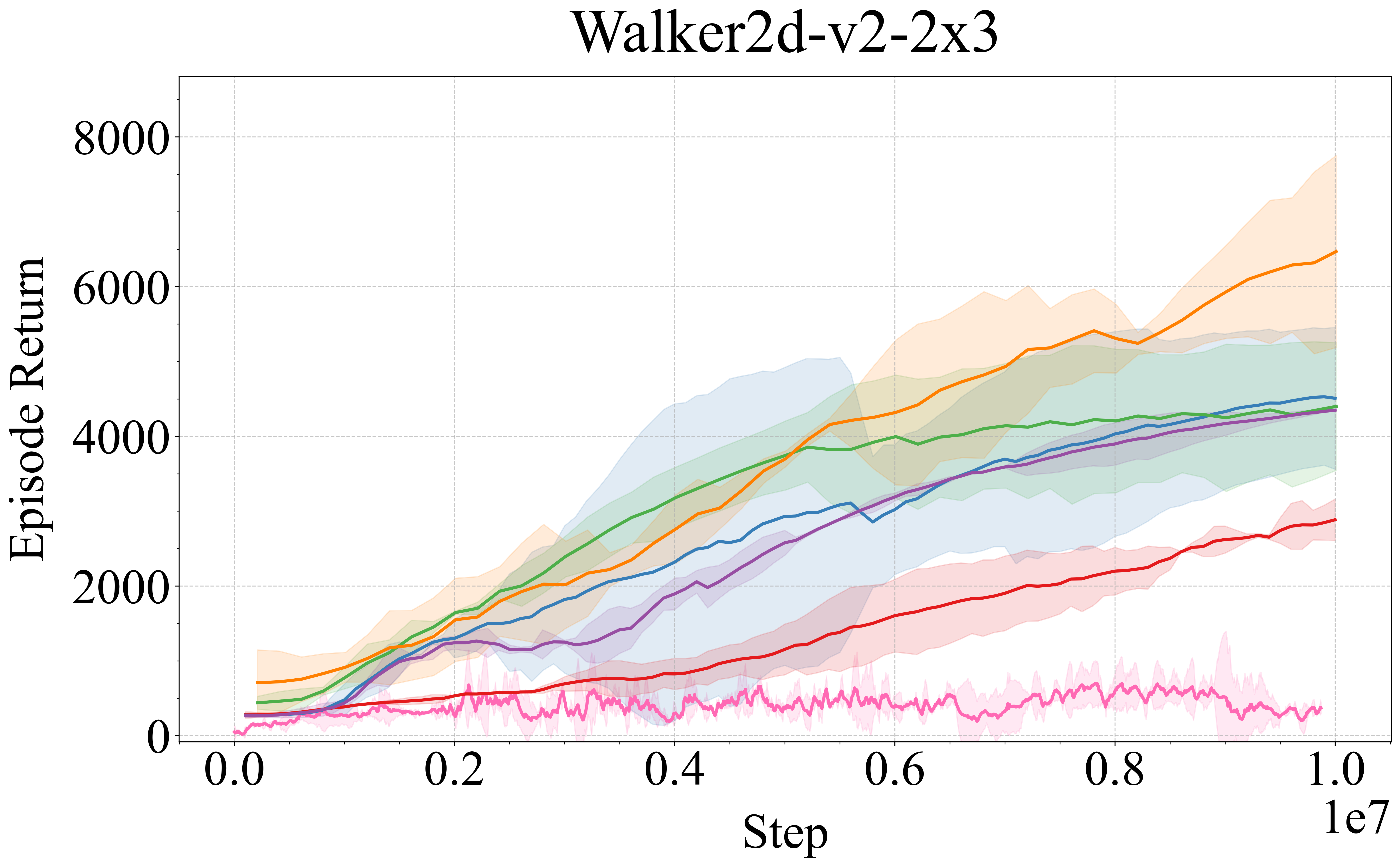}%
    }%
    \hfill%
    \subfloat[{6x1-Agent HalfCheetah}]{%
        \includegraphics[width=0.33\linewidth]{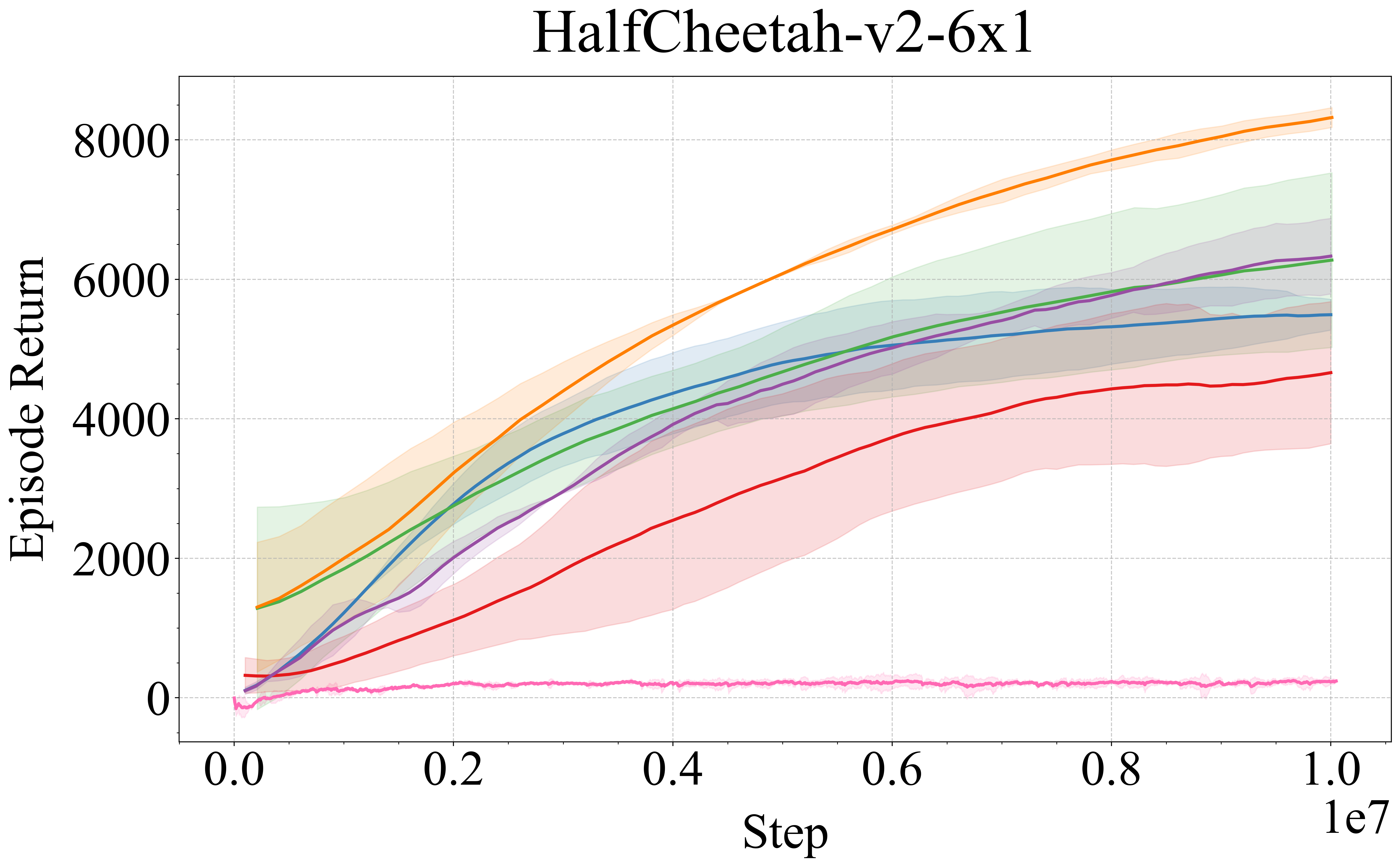}%
    }%
    \caption{The performance of algorithms on different tasks in multiple benchmark environments.}
    \label{fig:allFigs}
\end{figure*}

\subsubsection{Practical Algorithm}
To enhance the agent's awareness of its contribution to the policy while ensuring stable learning, we propose a hybrid credit assignment mechanism that guides the agent in learning effective strategies. In this framework, the agent is required to simultaneously maximize the global reward and the additional contribution derived from its own historical policies.
Therefore, we propose a new Historical Interaction-Enhanced Shapley Policy Gradient (HIS) to facilitate the learning of actor network parameters:
\begin{equation}
    \begin{split}
    &{\nabla_\theta}J_{\pi^{i_m}}(\theta^{i_m}) \\
    &= \mathbb{E}_{s_t \sim \mathcal{D}} \left[ \mathbb{E}_{{a}_t^{i_{1:m-1}} \sim \pi_{\theta_{new}^{i_{1:m-1}}}^{i_{1:m-1}}, {a}_t^{i_m} \sim \pi_{\theta^{i_m}}^{i_m}} \bigg[ \right. \\
    &\quad \nabla_\theta \Big(  Q_{\pi_{\text{old}}:\psi}^{i_{1:m}} \left( s_t, {a}_t^{i_{1:m-1}}, {a}_t^{i_m} \right)  - \alpha \log \pi_{\theta^{i_m}}^{i_m} \left( {a}_t^{i_m} \mid s_t \right) \Big) \bigg] \Bigg]  \\
    &\quad + \mathop \mathbb{E} \limits_{ \left({s}_t, {a}_t\right) \sim {\cal{D}}} \left[ \nabla_\theta  {\cal BC} \Big( \log {\pi _\theta } \left({a}^{i_m}_t|s_t \right) \Big)   Q_{\psi}^{\Phi_{i_m}}\left( s_t, {a}_t^{i_m} \right) \right].
    \end{split}
\end{equation}
where $ i _ { 1 : n } \in S y m ( n )$ is a randomly selected large coalition arrangement, $m = 1, \ldots, n$, and the policy networks of different agents are trained in a sequential updating manner \cite{liu2023maximum}.
HIS effectively balances global collaboration and individual credit assignment, thereby ensuring that the agent can optimize its strategy to maximize both global and self-interests while stabilizing training.

For training the critic network, we do not replace the action-value function with the sum of the agents' local credit distributions, as done in SQDDPG. 
Instead, we argue that directly optimizing the parameters of the centralized critic network using the temporal-difference (TD) error allows it to effectively evaluate the value functions corresponding to different coalitions. This approach maintains the consistency of the critic's evaluation while avoiding the additional complexity introduced by decomposing the credit assignments at the critic level.
Therefore, we optimize the central critic network with the following loss function:
\begin{equation}
    \begin{split}
    J_Q(\psi ) &= \mathbb{E}_{(s_t, {a}_t) \sim \mathcal{D}} \Bigg[ \frac{1}{2} \Bigg( r(s_t, {a}_t) + \gamma \mathbb{E}_{s_{t+1} \sim P} \left[ V_{\bar{\psi}}(s_{t+1}) \right] \\
    &\qquad - Q_{\psi}(s_t, {a}_t) \Bigg)^2 \Bigg].
    \end{split}
\end{equation}
where $\bar{\psi}$ represents the parameters of the target critic network.
The soft value function is obtained by implicitly parameterizing the soft Q function parameters via the following equation:
\begin{equation}
    \begin{split}
        V(s_{t}) = \mathbb{E}_{{a}_{t} \sim \pi} \left[ Q(s_{t}, {a}_{t}) + \alpha \sum_{i=1}^{n}{\mathcal{H}\left(\pi^i\left(\cdot^i | s\right)\right)} \right] \\
    \end{split}   
\end{equation}

For the temperature coefficient $\alpha$, we used the same automatic temperature coefficient adjustment method as SAC, and its optimization objective is:
\begin{equation}
J(\alpha) = \mathbb{E}_{\mathrm{s}_t \sim \mathcal{D}, {a}_t \sim \pi} \big[ -\alpha \log \pi({a}_t \vert \mathrm{s}_t) - \alpha \bar{\mathcal{H}} \big].
\end{equation}
where $\bar{\mathcal{H}}$ is the desired minimum expected entropy.
Specific implementation details and pseudocode can be found in Appendix C.
\section{Experiments}
\begin{figure*}[htbp] 
    \centering
    \subfloat[{2x4-Agent Ant-v2}]{%
        \includegraphics[width=0.33\linewidth]{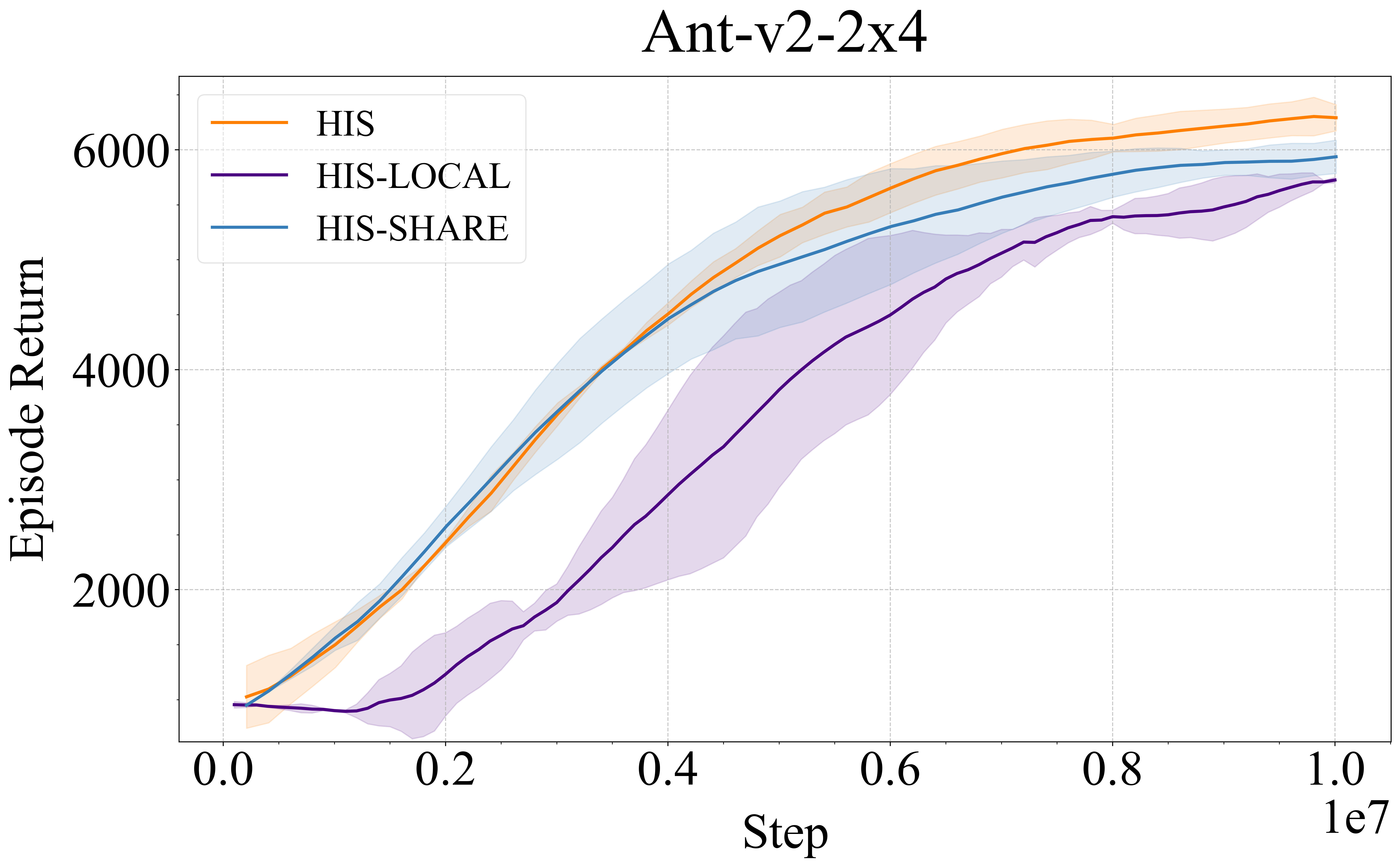}%
    }%
    \hfill%
    \subfloat[{8x1-Agent Ant-v2}]{%
        \includegraphics[width=0.33\linewidth]{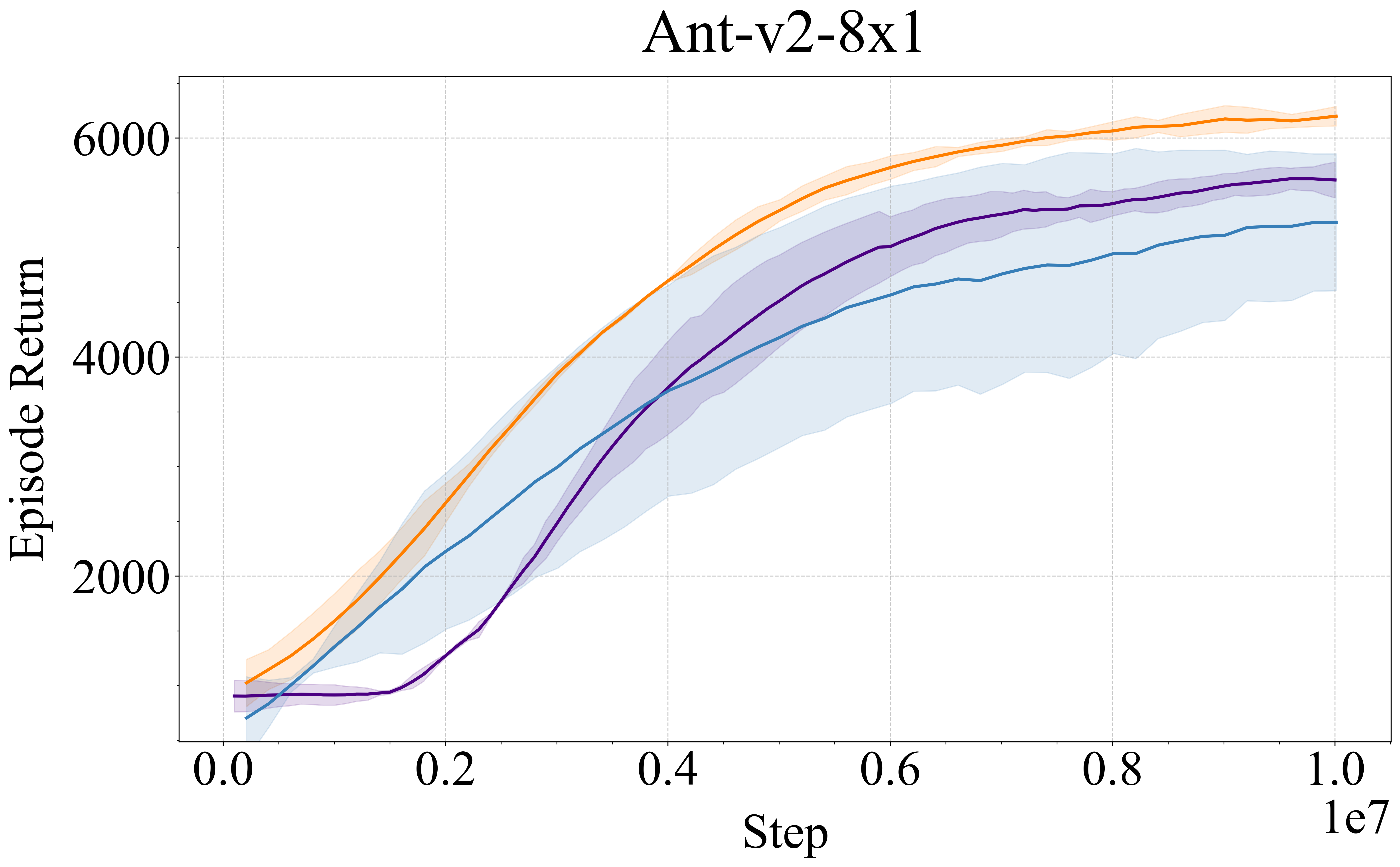}%
    }%
    \hfill%
    \subfloat[{Door Open Outward}]{%
        \includegraphics[width=0.33\linewidth]{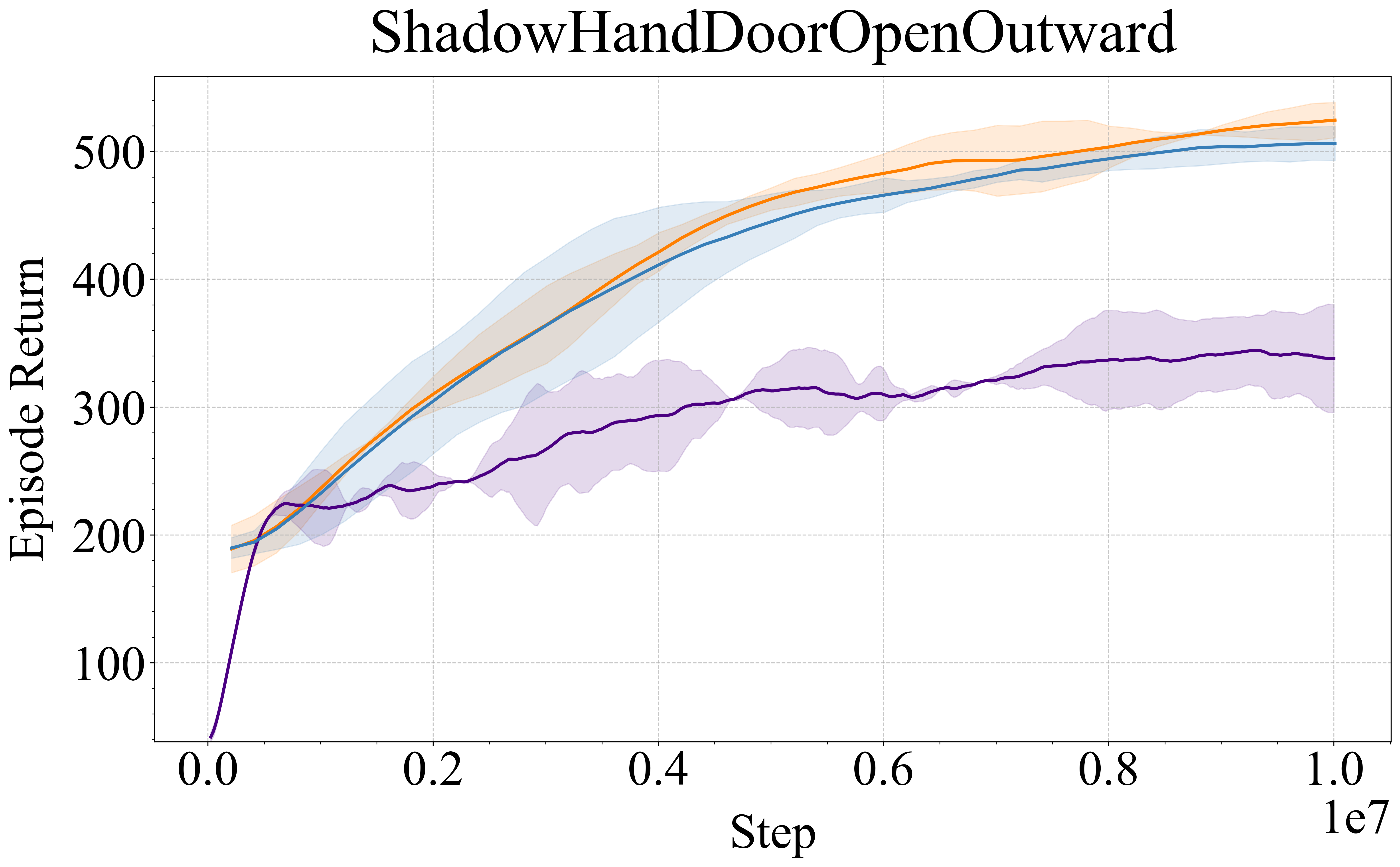}%
    }
    \caption{Ablation study of different HIS components on Bi-DexHands and MAMuJoCo tasks.}
    \label{fig:ablation}
\end{figure*}

\subsection{Experimental Setup}
To thoroughly evaluate the effectiveness of the HIS algorithm, we tested it across three widely-used continuous action benchmark environments: Multi-Agent Particle Environment (MPE) \cite{lowe2017multi}, Bi-DexHands \cite{chen2022towards}, and Multi-Agent MuJoCo (MAMuJoCo) \cite{de2020deep}. 
Our experiments employed state-of-the-art multi-agent reinforcement learning (MARL) algorithms, such as HAPPO, HATRPO \cite{kuba2021trust}, MAPPO \cite{yu2022surprising}, HASAC \cite{liu2023maximum}, and FACMAC \cite{peng2021facmac}, as baselines, and compared their performance with HIS across five random seeds.
Through these experiments, we aim to demonstrate that HIS offers advantages in the following areas: 
\textbf{(1)} Compared with the shared reward method, HIS can enhance the agent's ability to perceive its own behavior through a hybrid assignment mechanism.
\textbf{(2)} In contrast to the local reward method, HIS exhibits significantly more stable performance.
\textbf{(3)} The design of HIS using historical action data makes it more sample efficient than existing algorithms.

\subsection{Experimental Results}
We conducted extensive experiments to evaluate the performance of HIS across multiple environments. Specifically, we tested HIS on three tasks within the MPE environment and nine tasks each in the Bi-DexHands and MAMuJoCo environments. Experimental results show that HIS outperforms the baselines and we achieve the best performance on a total of 21 tasks.
Full experimental details as well as hyperparameter designs can be found in Appendix D.

\textbf{MPE}: 
It is important to note that, to align with fully cooperative multi-agent reinforcement learning, we aggregated the individual rewards of each agent within these environments and returned only the global reward during training.
As shown in Figure \ref{fig:allFigs}, the HIS algorithm exhibits a faster convergence rate compared to the other algorithms, and it also demonstrates superior stability after reaching convergence. This indicates that HIS not only accelerates learning but also ensures consistent and reliable performance across different tasks.

\textbf{MAMuJoCo}: 
In our experiments, we included FACMAC as an additional baseline algorithm. FACMAC \cite{peng2021facmac} uses a local reward scheme and applies the same implicit assignment method as QMIX \cite{rashid2020monotonic} to address the credit assignment problem through value decomposition. However, as shown in Figure \ref{fig:allFigs}, although FACMAC can implicitly allocate agent contributions, the accuracy of value decomposition is not guaranteed in tasks with strong inter-agent coupling.
In contrast, the HIS algorithm, which combines both global rewards and local credit assignment rewards, demonstrates significantly improved stability and robustness, particularly in complex collaborative tasks. 
Notably, HIS significantly outperforms the shared reward scheme by introducing the hybrid credit assignment mechanism, which is especially effective in tasks involving a large number of agents.
The shared reward mechanism can easily obscure individual contributions, while the hybrid credit assignment mechanism of HIS can help the agent more clearly understand the impact of its historical strategies on the global results, thereby optimizing the decision-making process.

\textbf{Bi-DexHands}: 
As shown in Figure \ref{fig:allFigs}, the HIS algorithm benefits from the additional contribution incentives provided by its hybrid credit assignment mechanism, which motivates the agent to explore strategies that maximize its own benefits. 
Therefore, HIS outperforms other baseline algorithms in these robotic arm manipulation tasks, demonstrating both strong sample efficiency and stability.
The algorithm enhances the ability to optimize individual contributions, accelerates learning, and ensures consistent performance in challenging tasks that require fine-tuning. These results underscore the dual advantages of the hybrid credit assignment mechanism: it not only accelerates learning by offering more informative feedback but also encourages the development of resilient policies that maintain high performance even in complex scenarios. 
This highlights the importance of designing effective credit assignment mechanisms when advancing multi-agent reinforcement learning in complex and high-precision tasks.

\subsection{Ablation Study}
To thoroughly investigate the impact of hybrid credit assignment on algorithm performance, we conduct a series of ablation experiments in three tasks.
The experiments focused on two primary aspects: 
\textbf{(1)} whether the hybrid credit assignment scheme enhances performance compared to the shared reward scheme, and \textbf{(2)} whether the hybrid credit assignment scheme improves stability and overall effectiveness when compared to the local reward scheme.

In response to these two points, we proposed two variants: HIS-SHARE and HIS-LOCAL. HIS-SHARE replaces the calculated Shapley Q-value in policy learning with a soft Q-function representing the global reward. HIS-LOCAL uses the same modeling method as SQDDPG, which directly replaces the soft Q-function in policy learning with the calculated Shapley Q-value of each agent.


As shown in Figure \ref{fig:ablation}, HIS-SHARE assigns the same reward to all agents, and a single agent cannot obtain the contribution incentive related to its own strategy, resulting in a significant performance degradation.
This underscores the importance of the hybrid credit assignment mechanism, which guides agents toward policies that not only contribute to global objectives but also reflect their individual impact. Without local incentives, agents tend to converge on suboptimal behaviors that prioritize collective outcomes at the expense of individual contributions, a situation particularly detrimental in strongly coupled or highly interdependent tasks.

The HIS-LOCAL algorithm, by contrast, is overly dependent on the optimization properties of the trust assignment scheme, which limits its adaptability to more complex tasks. This dependence leads to suboptimal performance in strongly coupled scenarios, where accurate calculation of the Shapley value is difficult. Furthermore, the need to repeatedly compute the Shapley value introduces significant time complexity, imposing a substantial computational burden during training. This not only slows down the learning process but also limits the scalability of the algorithm in large-scale environments.

\section{Conclusion}
In this paper, we introduce the Historical Interaction-Enhanced Shapley Policy Gradient Algorithm (HIS) for Multi-Agent Credit Assignment, which integrates a hybrid credit assignment mechanism and theoretically demonstrates that the output results of this mechanism are efficient and stable.
HIS leverages historical interaction information to calculate each agent's contribution to the overall task when computing the policy gradient, thereby enhancing the agent's local perception of its own policy in a highly sample-efficient manner.
We evaluate the performance of HIS in three widely-used continuous benchmark environments. 
The experimental results demonstrate that, compared to the most widely adopted state-of-the-art algorithms, HIS offers superior robustness and adaptability. 
Notably, the hybrid credit assignment mechanism enables our algorithm to excel in complex and highly coupled tasks, where other methods may struggle to achieve stable performance.
\bibliography{aaai2026}
\newpage
\clearpage

\appendix

\section{Definitions and Theorems}
\subsection{Convex Game}
\begin{definition}\label{ap-definition-1}
    A transferable utility game $( \mathcal{N} , v )$ is convex if for $ \forall \mathcal{C} , \mathcal{D} \subseteq N$ the following condition holds:
    \begin{equation} 
        v ( \mathcal{C} \cup \mathcal{D} ) + v ( C \cap D  ) \geq v ( C  ) + v ( D )
    \end{equation}
\end{definition}
\begin{definition} \label{ap-definition-2}
    If $x(\mathcal{C}) = v(\mathcal{C})$ for any $\cal {C} \in \cal{CS}$ in an outcome $(\mathcal{CS}, x)$, then this outcome is efficient and maximazes the social welfare and maximizes the social welfare.
\end{definition}
\begin{definition} \label{ap-definition-3}
    Given a convex game $(\mathcal{N},v)$, the Core is the set of all stable solutions. For each $\mathcal{C} \in \mathcal{N}$, there is $ x ( C ) \geq v ( C )$, where $ x ( \mathcal{C} ) = \sum _ { i \in \mathcal{C} } x ^ { i }$.
\end{definition}
\begin{definition}
    The Core of a convex game must be non-empty.
\end{definition}

\subsection{Shapley Value}
\begin{theorem}
    Given a convex game $G=(\mathcal{N},v)$, the Shapley value must belong to the Core, and the assignment scheme is stable.
\end{theorem}
\section{Proofs}
\begin{figure*}[t] 
    \centering
    \subfloat[{Catch Abreast}]{%
        \includegraphics[width=0.33\linewidth]{figs/performance/dexhands/ShadowHandCatchAbreast.png}%
    }
    \hfill%
    \subfloat[{Two Catch Underarm}]{%
        \includegraphics[width=0.33\linewidth]{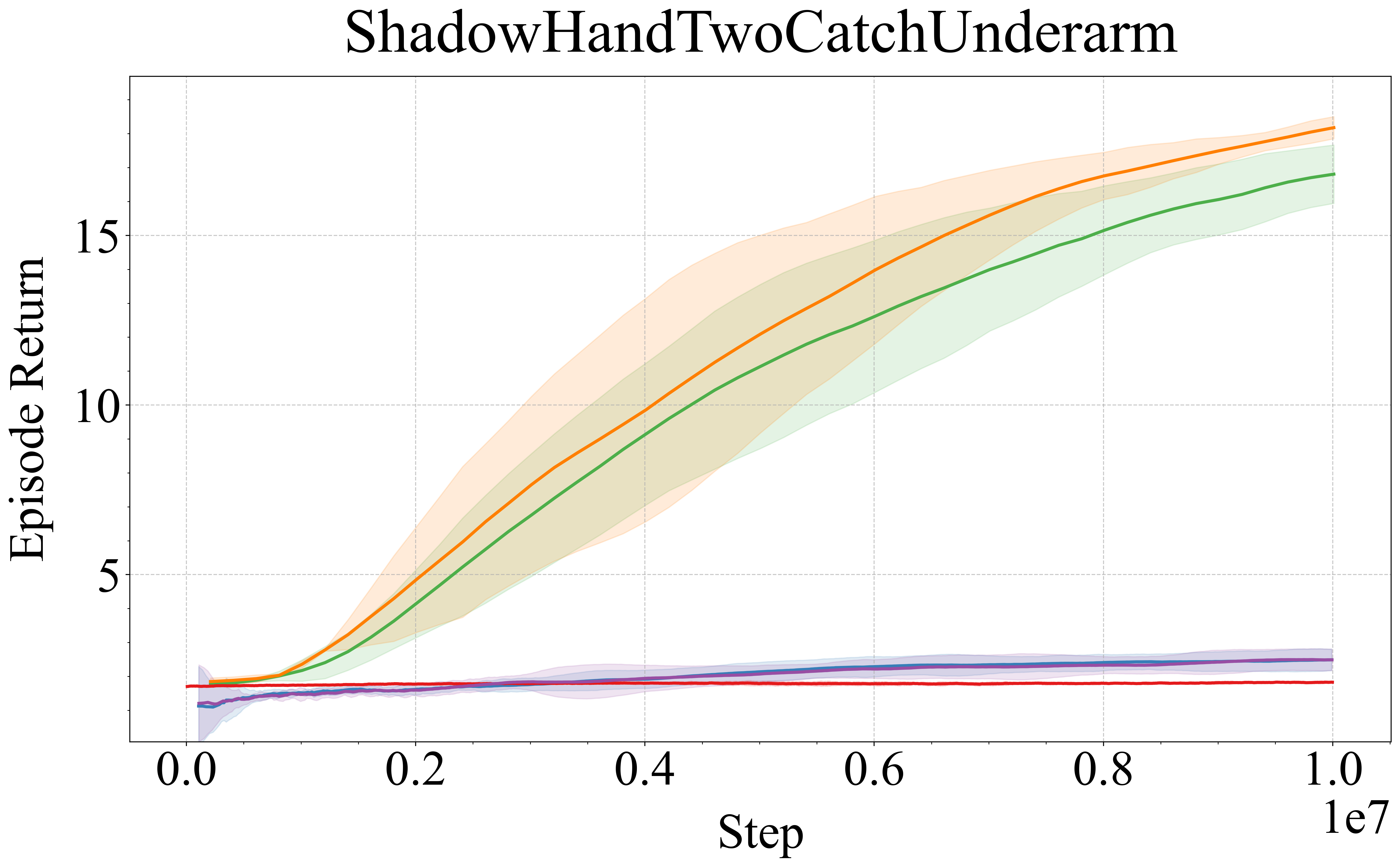}%
    }%
    \hfill%
    \subfloat[{Catch Over2Underarm}]{%
        \includegraphics[width=0.33\linewidth]{figs/performance/dexhands/ShadowHandCatchOver2Underarm.png}%
    }\\[1ex]
    \subfloat[{Door Close Inward}]{%
        \includegraphics[width=0.33\linewidth]{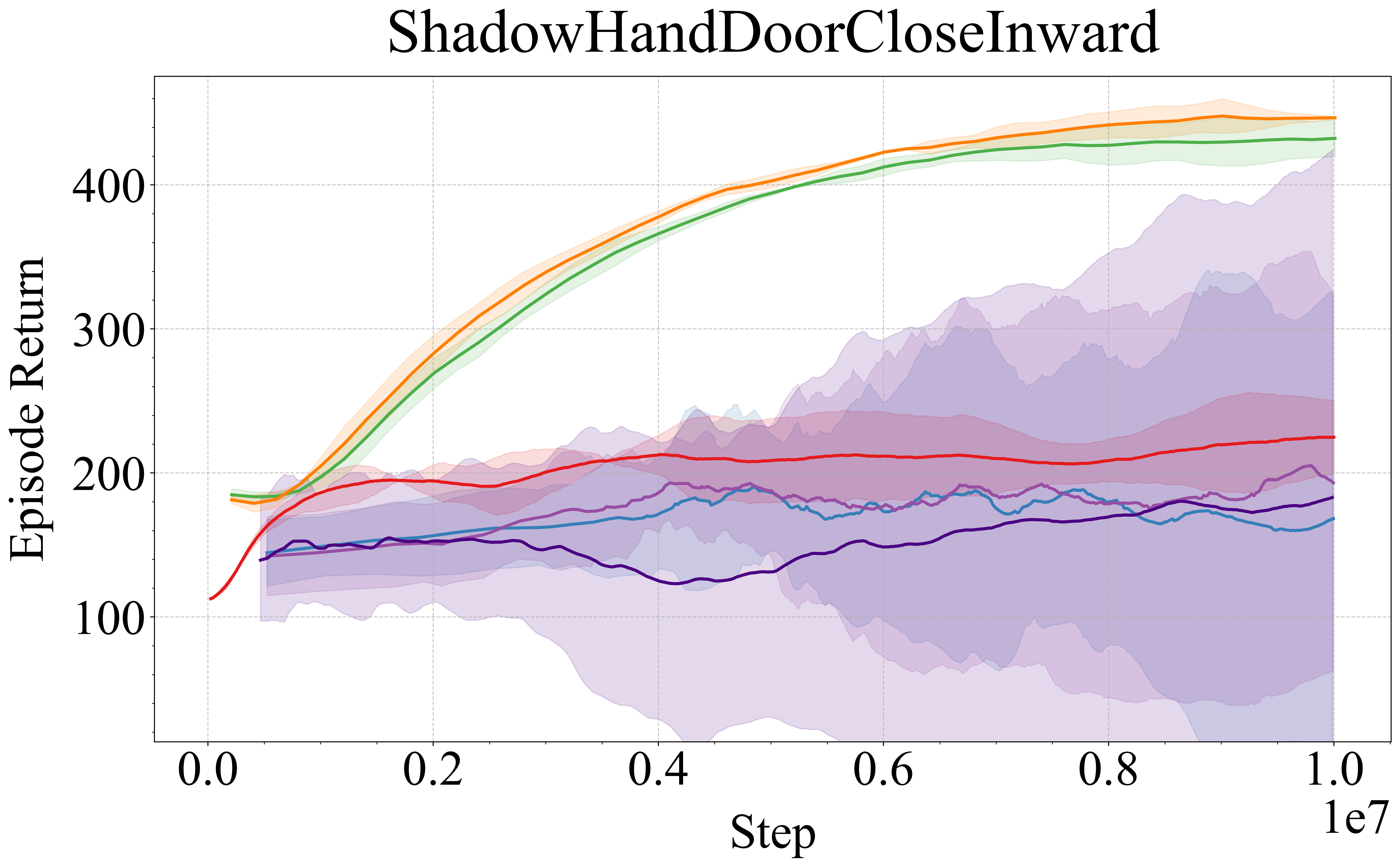}%
    }
    \hfill%
    \subfloat[{Door Open Inward}]{%
        \includegraphics[width=0.33\linewidth]{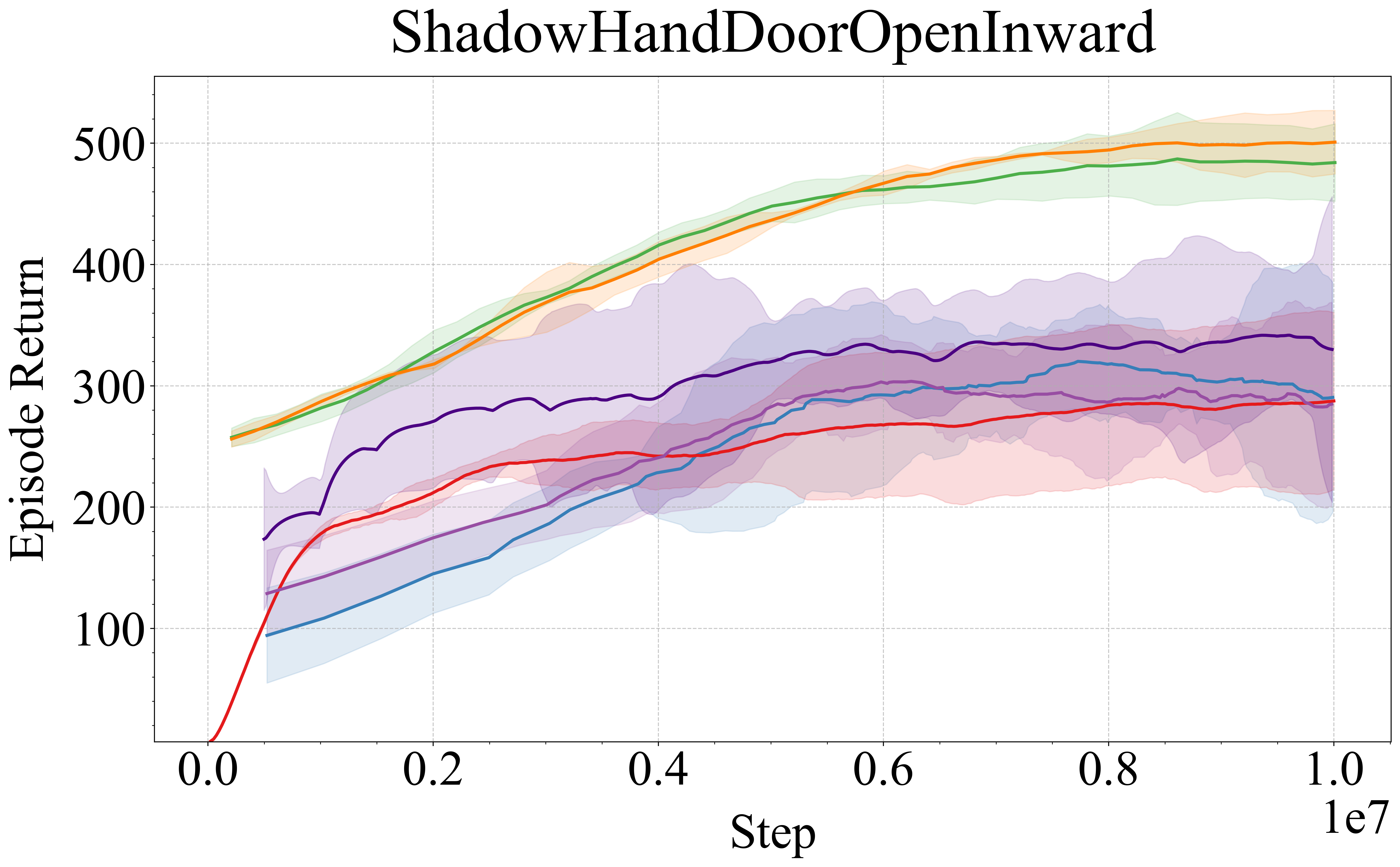}%
    }%
    \hfill%
    \subfloat[{Door Open Outward}]{%
        \includegraphics[width=0.33\linewidth]{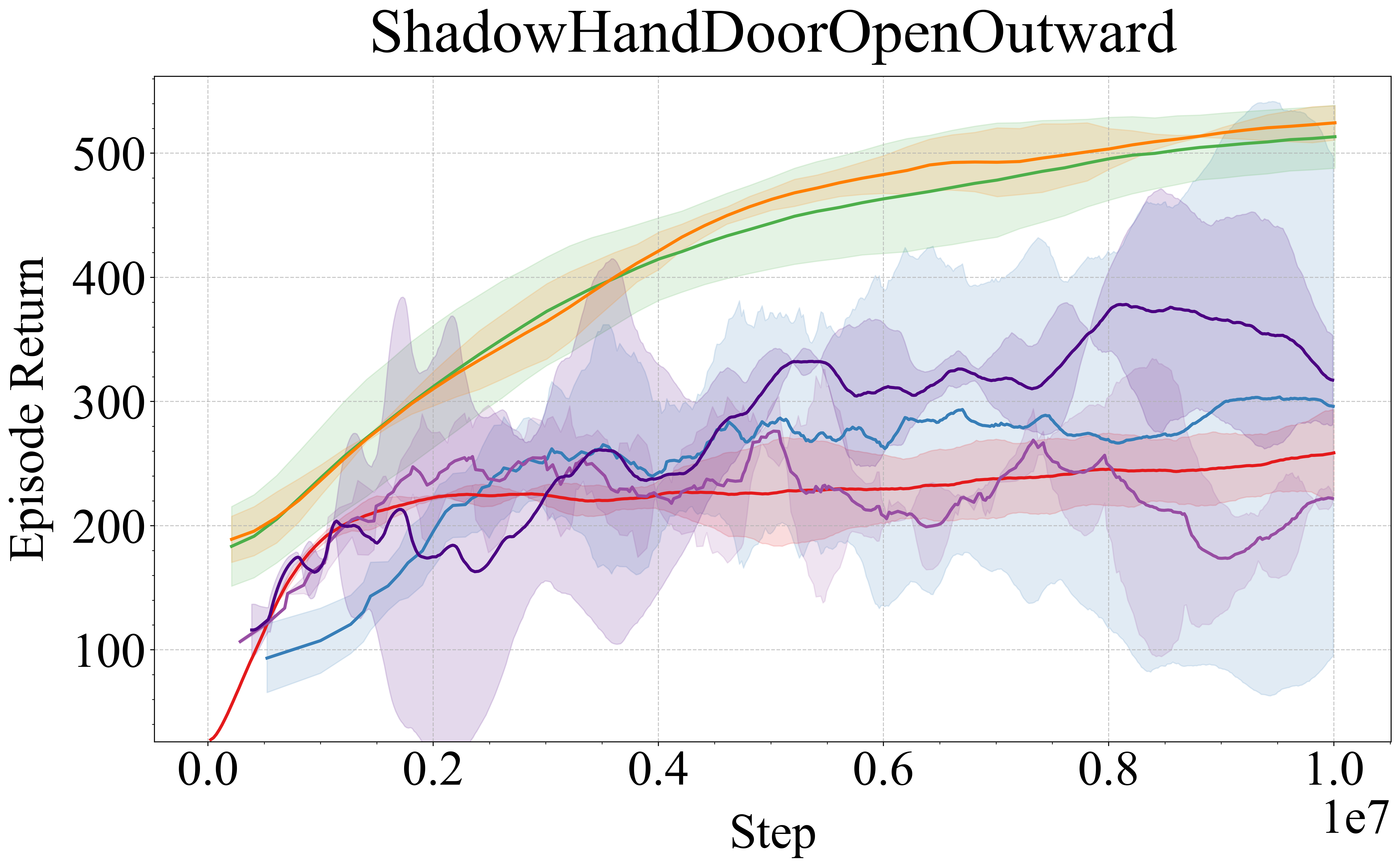}%
    }\\[1ex]
    \subfloat[{Lift Underarm}]{%
        \includegraphics[width=0.33\linewidth]{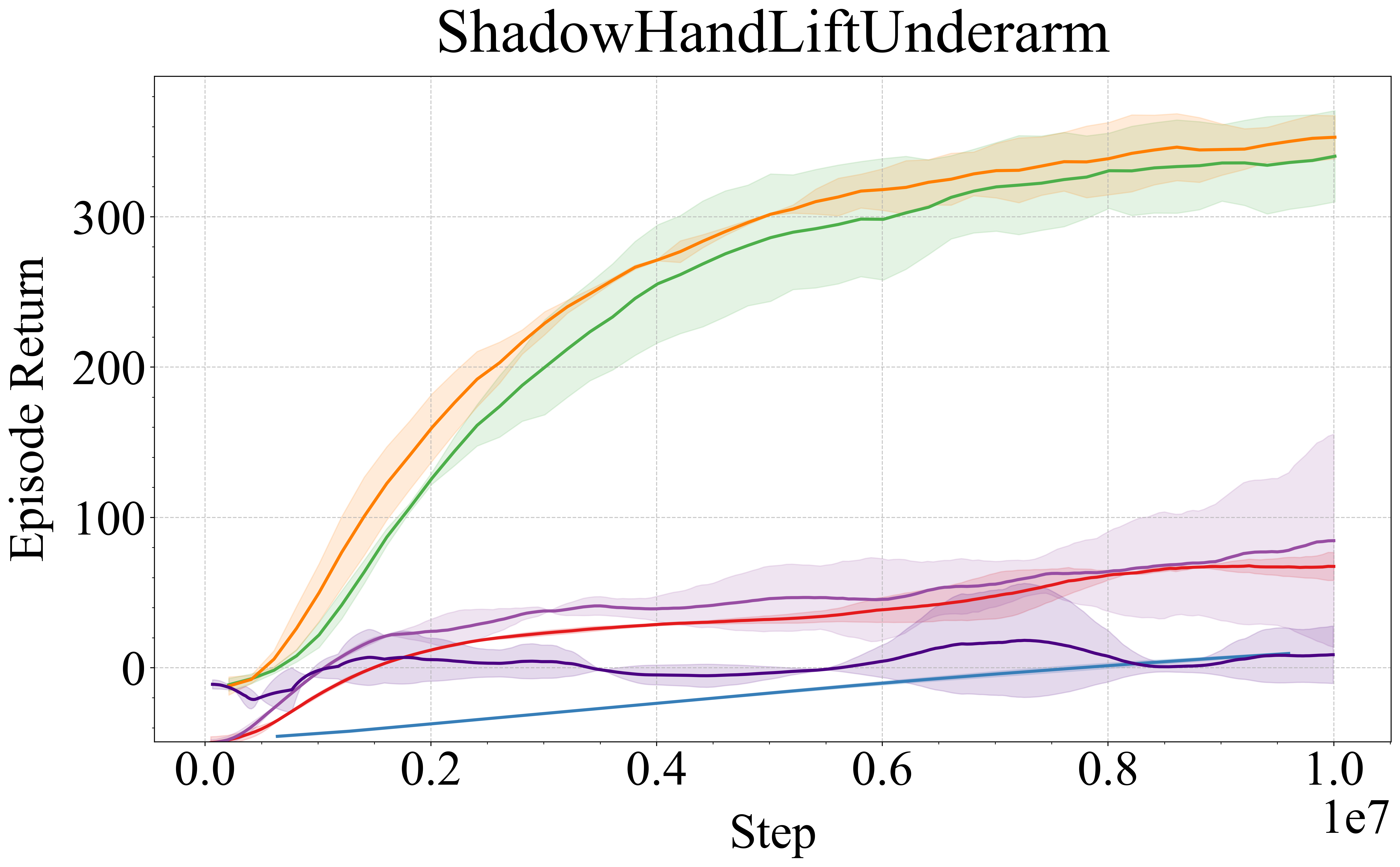}%
    }
    \hfill%
    \subfloat[{Hand Over}]{%
        \includegraphics[width=0.33\linewidth]{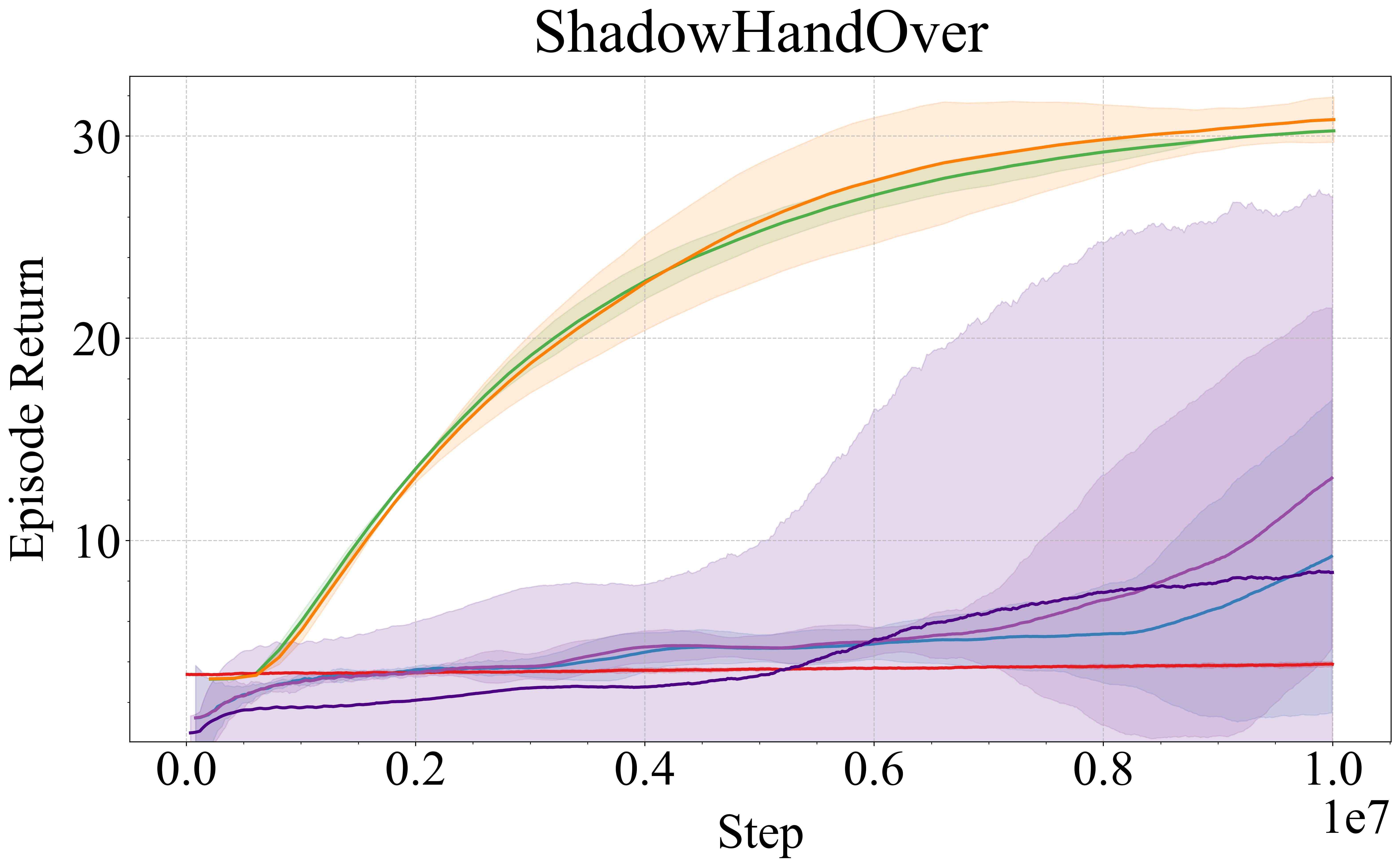}%
    }%
    \hfill%
    \subfloat[{Hand Pen}]{%
        \includegraphics[width=0.33\linewidth]{figs/performance/dexhands/ShadowHandPen.png}%
    }%
    \caption{Performance of algorithms in different tasks in Bi-DexHands.}
    \label{fig:DexHands}
\end{figure*}
\begin{figure*}[t] 
    \centering
    \subfloat[{2x4-Agent Ant}]{%
        \includegraphics[width=0.33\linewidth]{figs/performance/mamujoco/Ant-v2-2x4.png}%
    }
    \hfill%
    \subfloat[{4x2-Agent Ant}]{%
        \includegraphics[width=0.33\linewidth]{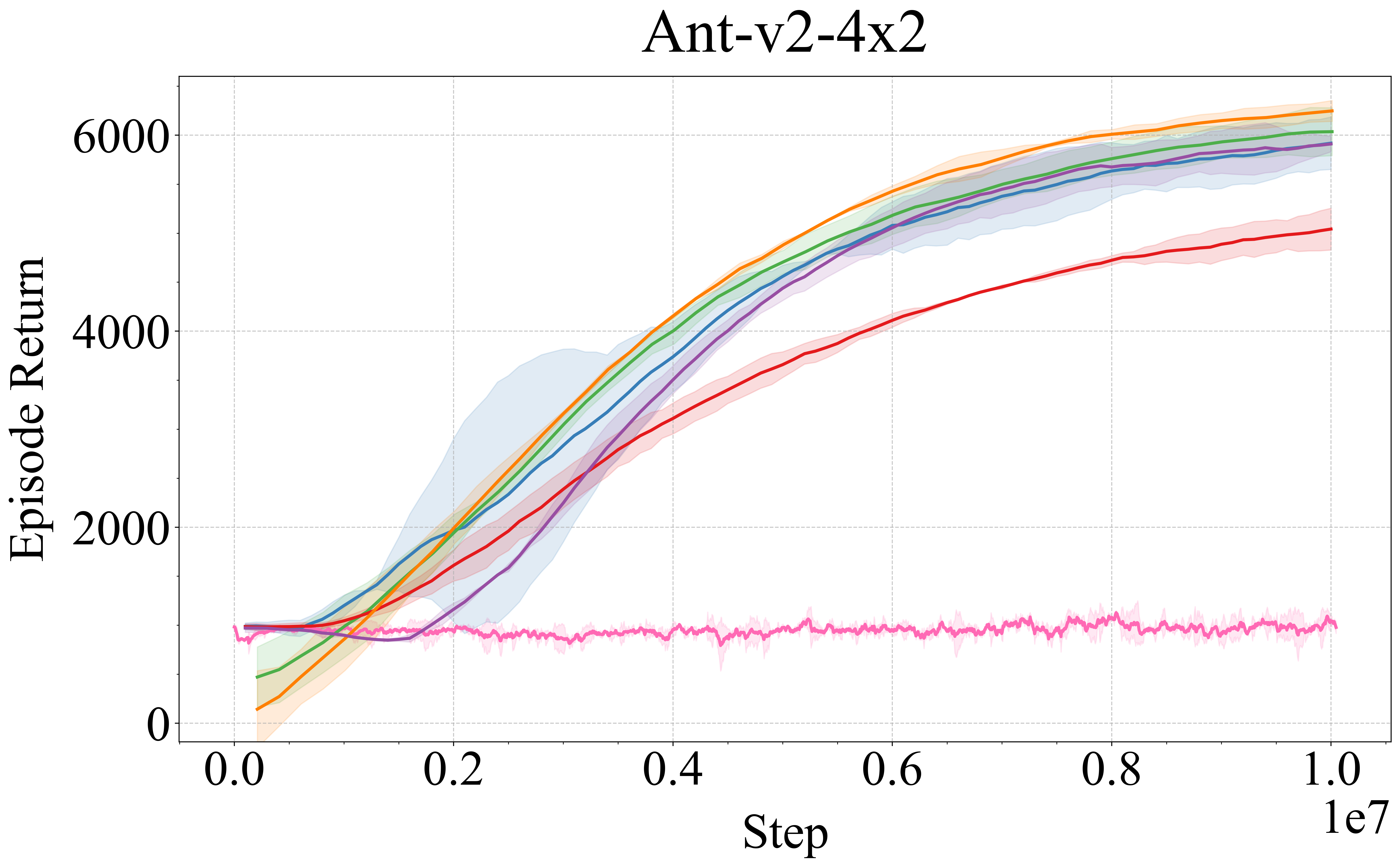}%
    }%
    \hfill%
    \subfloat[{8x1-Agent Ant}]{%
        \includegraphics[width=0.33\linewidth]{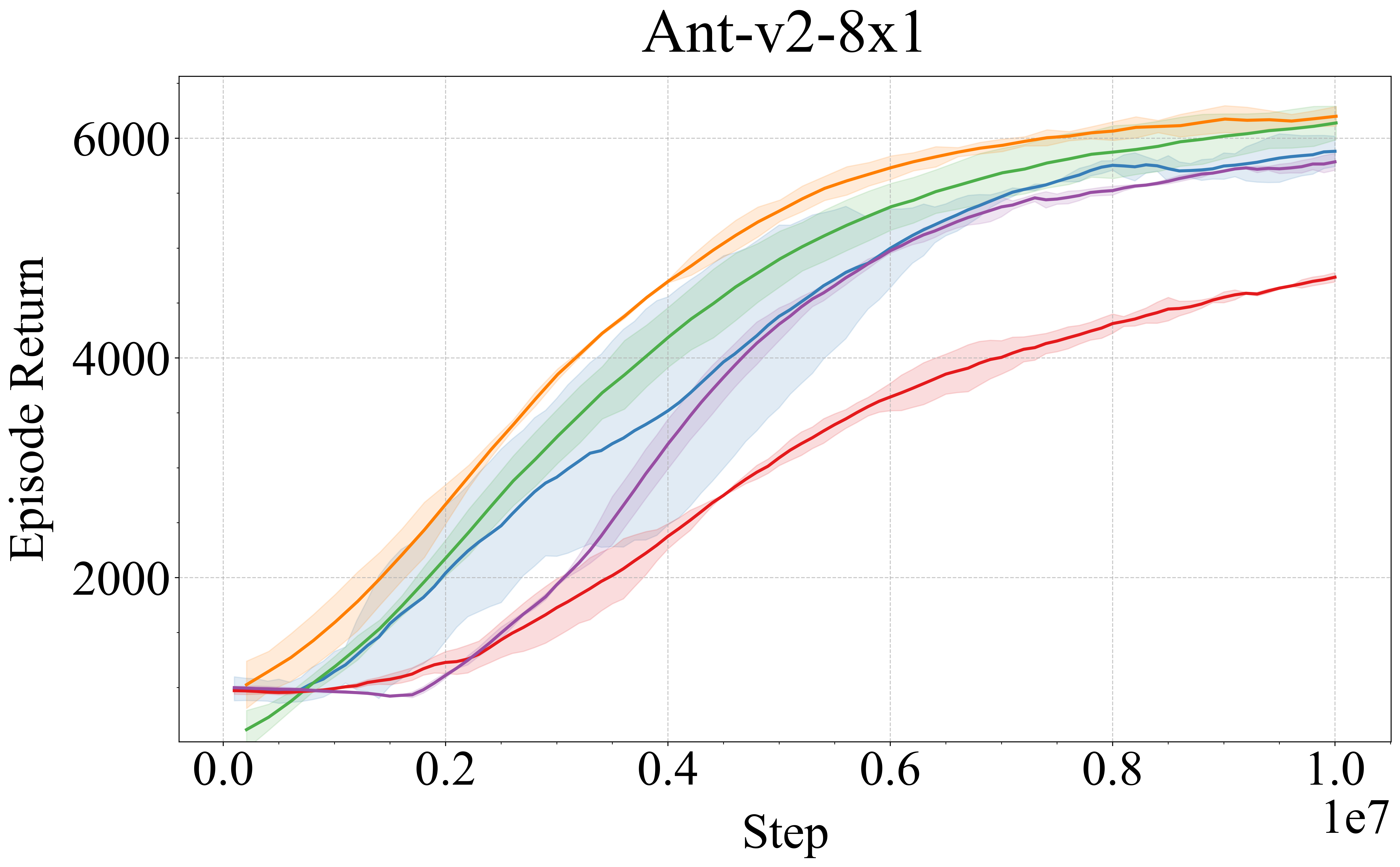}%
    }\\[1ex]
    \subfloat[{2x3-Agent HalfCheetah}]{%
        \includegraphics[width=0.33\linewidth]{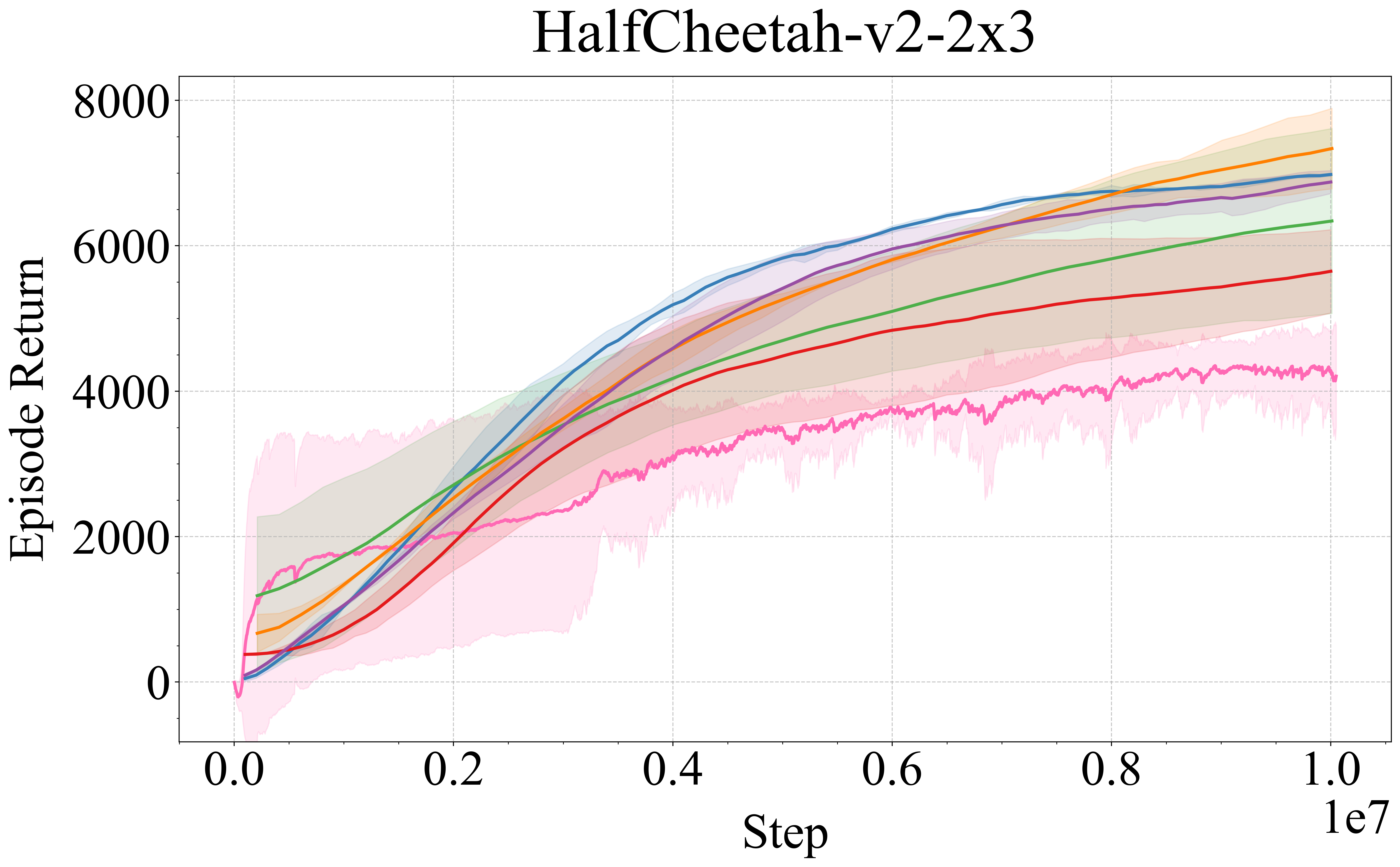}%
    }%
    \hfill%
    \subfloat[{3x2-Agent HalfCheetah}]{%
        \includegraphics[width=0.33\linewidth]{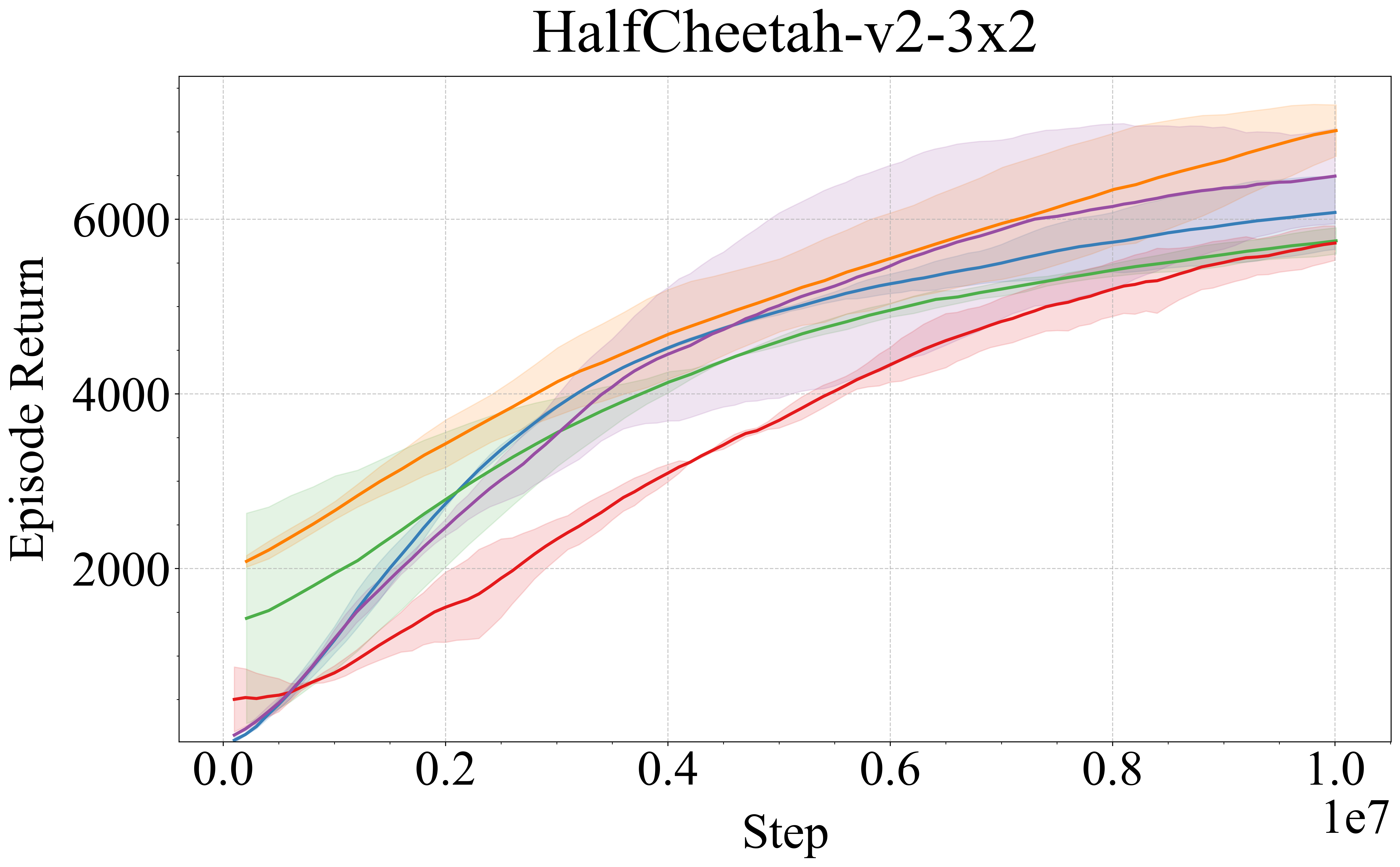}%
    }%
    \hfill%
    \subfloat[{6x1-Agent HalfCheetah}]{%
        \includegraphics[width=0.33\linewidth]{figs/performance/mamujoco/HalfCheetah-v2-6x1.png}%
    }\\[1ex]
    \subfloat[{10x2-Agent manyagent\_swimmer}]{%
        \includegraphics[width=0.33\linewidth]{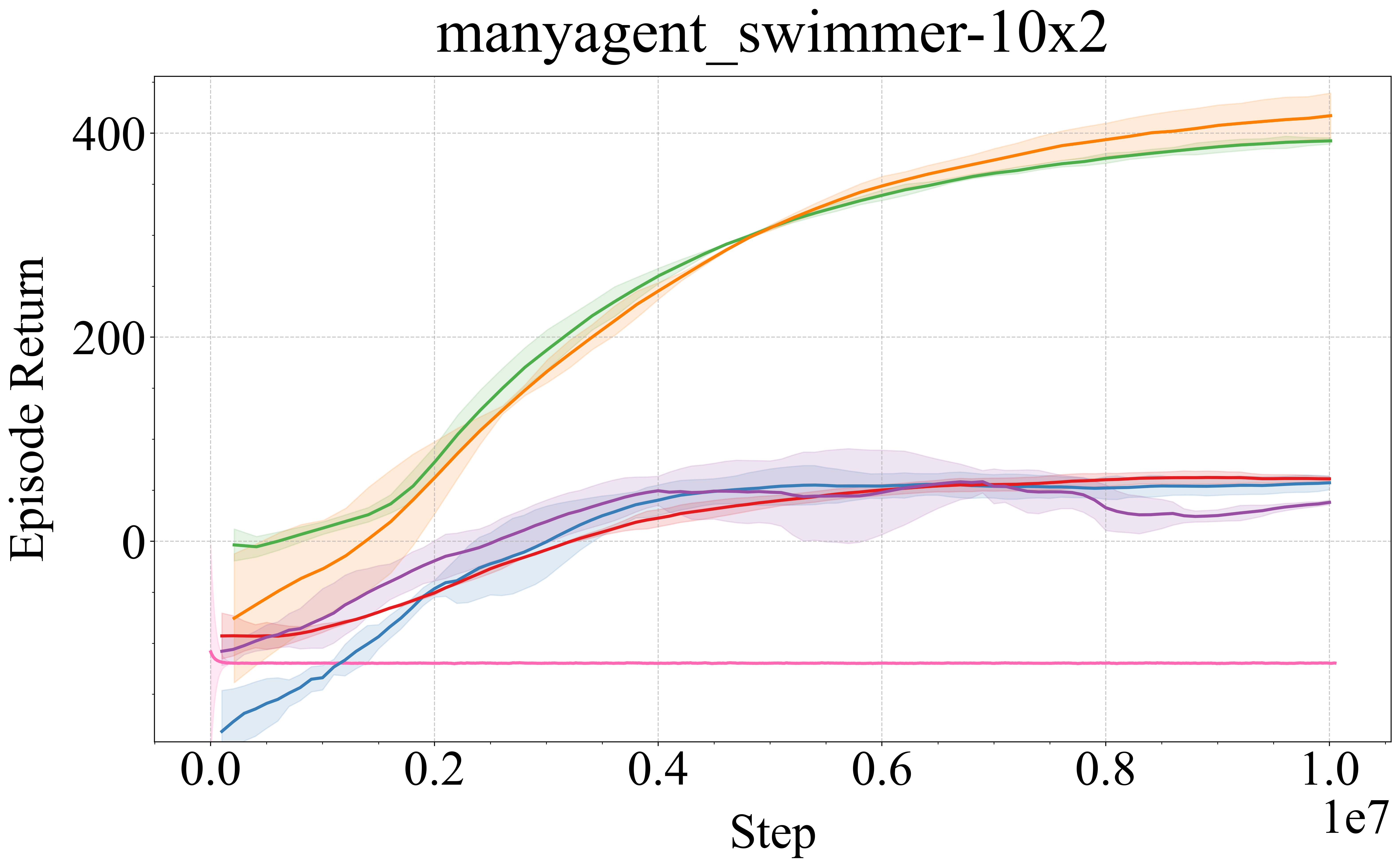}%
    }%
    \hfill%
    \subfloat[{2x3-Agent Walker2d}]{%
        \includegraphics[width=0.33\linewidth]{figs/performance/mamujoco/Walker2d-v2-2x3.png}%
    }%
    \hfill%
    \subfloat[{6x1-Agent Walker2d}]{%
        \includegraphics[width=0.33\linewidth]{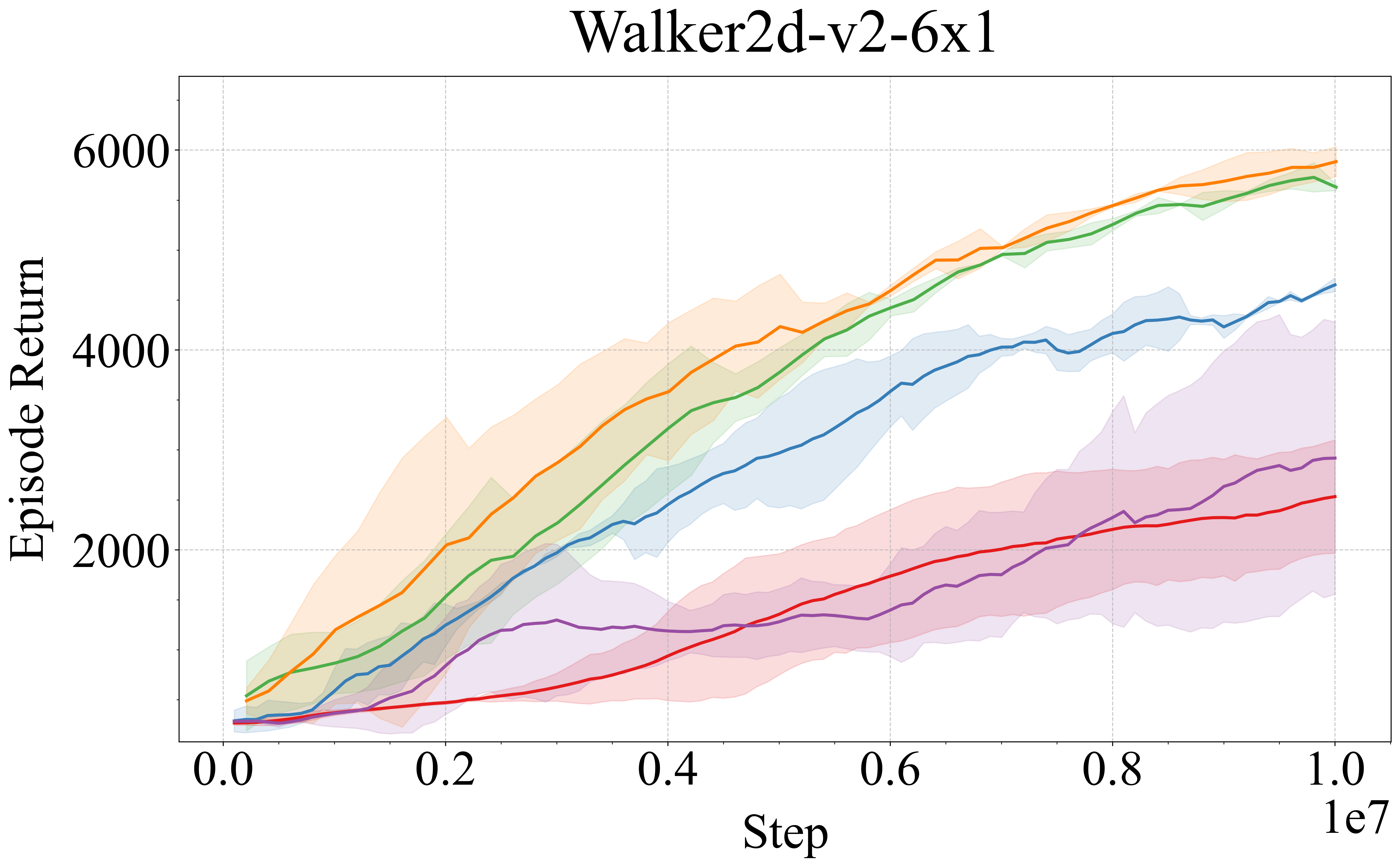}%
    }%
    \caption{Performance of algorithms in different tasks in MAMuJoCo.}
    \label{fig:MuJoCo}
\end{figure*}

\subsection{Proof of Theorem 2}
\begin{theorem} \label{theorem-ap-2}
    Given a convex game $G=(\mathcal{N},v)$, the hybrid assignment outcome $(\mathcal{N},x)$ is efficient, and the payoff vector $x$ is an efficient reward reallocation.
\end{theorem}

\begin{proof}
    $x(\mathcal{N})$ is defined as the sum of the payoff vector under the grand coalition $\mathcal{N}$:
    \begin{equation}
        \begin{split}
            x(\mathcal{N}) &= \sum_{i\in \mathcal{N}}x^{i}  \\
            &= \sum_{i\in \mathcal{N}}\sum_{\mathcal{C}\subseteq N\setminus\{i\}} \frac{|\mathcal{C}|!(n-|\mathcal{C}|-1)!}{n!}\left[ \frac{v(\mathcal{C}\cup\{i\})}{2} \right. \\
            & \quad \left. - \frac{v(\mathcal{C})}{2} \right] + \sum_{i\in \mathcal{N}}\frac{v(\mathcal{N})}{2|\mathcal{N}|}.
        \end{split}
    \end{equation}
    Since $v(\emptyset)=0$, and it only appears once with $v(\mathcal{N})$ in the summation process, the value of the coalition $v(\mathcal{C})$ consisting of $p$ agents will appear $n-p$ times in the form of a negative sign and $p$ times in the form of a positive sign, which can be expanded as follows:
    \begin{equation}
        \begin{split}
        \sum_{i\in N}x^{i} &= n \frac{(n - 1)!1!}{n!} \frac{v(\mathcal{N})}{2} + n \frac{0!(n - 1)!}{n!} \frac{v(\emptyset)}{2} \\
        &\quad + \sum_{\substack{\mathcal{C} \subset \mathcal{N}, \mathcal{C} \neq \emptyset, \\ |\mathcal{C}| = p}} \left( p \frac{(p - 1)!(n - p)!}{n!} - \right. \\
        &\quad \left. (n - p) \frac{p!(n - p - 1)!}{n!} \right) \frac{v(\mathcal{C} )}{2} + \frac{v(\mathcal{N})}{2} \\
        &= \frac{v(\mathcal{N})}{2}+\frac{v(\mathcal{N})}{2}\\
        &= v(\mathcal{N}).
        \end{split}
        \tag{4}
    \end{equation}
    According to Definition \ref{ap-definition-2}, the output result $(N,x)$ of the hybrid assignment mechanism in the grand coalition convex game is efficient and can maximize social welfare.
\end{proof}
\subsection{Proof of Theorem 3}
\begin{theorem} \label{theorem-ap-3}
    Given a convex game $G=(\mathcal{N},v)$, the hybrid assignment outcome $(\mathcal{N},x)$ is in the Core, and the payoff vector $x$ is a stable reward reallocation.
\end{theorem}

\begin{proof}
    This paper allocates half of  $v(\mathcal{C})$ as the basis for computing the Shapley value.
    According to the definition of the Shapley value, we have:
    \begin{equation}
        \Phi_{i}(\mathcal{C})\!=\sum_{\mathcal{C}\subseteq\mathcal{N}\setminus\{i\}}\frac{|\mathcal{C}|!(|\mathcal{N}|-|\mathcal{C}|-1)!}{|\mathcal{N}|!}[\frac{v(\mathcal{C}\cup\{i\})\!-\!v(\mathcal{C})}2].
    \end{equation}
    Next, the marginal contribution can be used to define the payoff vector $\delta^{i}(\mathcal{C})=\frac{(v({\cal C}\cup\{i\})-v({\cal C}))}{2}$ for the Shapley value part.
    The reward vector of each agent is expressed as $x^{i}=(\frac{v(\mathcal{C}) }{2*k}+\delta^{i}(\mathcal{C}))$, where $k$ is the number of players in the coalition $\mathcal{C}$.

    Then, we will prove that the joint payoff vector $x$ of multiple agents is inside the Core. For any coalition ${\cal C}=\{i^{1},i^{2},\cdot\cdot\cdot,i^{k}\}\subseteq{\cal N}=\{1,...,n\}$, $v(\mathcal{C})$ can be expanded in series:
    \begin{equation}
        \begin{array}{c c c}{{v(\mathcal{C})=v(\{i^{1}\})-v(\emptyset)+v(\{i^{1},i^{2}\})}}\\ {{-v(\{i^{1}\})+\cdot\cdot\cdot+v(\mathcal{C})-v(\mathcal{C}\setminus\{i^{k}\}).}}\end{array}
    \end{equation}

    Without loss of generality, assume that $i^{1} < i^{2}< \ldots < i^{k}$. Let the coalition $\mathcal{D} =\{1,2,\ldots,i^{j-1}\}$, where $ j \in \left\{ 1 , \ldots , k \right\}$.
    According to Definition \ref{ap-definition-1}, we can get:
    \begin{equation}
        \begin{array}{l}{{v(\{i^{1},\ldots,i^{j-1},i^{j}\})-v(\{i^{1},\ldots,i^{j-1}\})}}\\ {{\leq\ v(\{1,\ldots,i^{j-1},i^{j}\})-v(\{1,\ldots,i^{j-1}\})}}\\ {{=\delta^{i^{j}}}}\\{{=v(\mathcal{D}\cup\{i^j\})-v(\mathcal{D})}.}\end{array}
    \end{equation}
    Divide the equation by 2 and then add $\frac{v(\mathcal{C}) }{2*k}$ to both sides to get:
    \begin{equation}
        \begin{array}{l}{\frac{v(\{i^{1},\ldots,i^{j-1},i^{j}\})-v(\{i^{1},\ldots,i^{j-1}\})}2}{+\frac{v(\mathcal{C}) }{2*k}}\\ {{\leq\frac{v(\mathcal{C}) }{2*k}+\frac{\delta ^{i^j}}2}}\\{{=\frac{v(\mathcal{C}) }{2*k}+\frac{v(\mathcal{D}\cup\{i^j\})-v(\mathcal{D})}2}}\\{{=\frac{v(\mathcal{C}) }{2*k}+\delta^{i^{j}}}}\\{{=x^{i^j}.}}\end{array}
    \end{equation}
    Then let $j$ from 1 to $k$ to sum the equation to get:
    \begin{equation}
        v({\mathcal{C}})\leq\sum_{j=1}^{k}x^{i^{j}}\,=\,x({\mathcal{C}}).
    \end{equation}
    According to Definition \ref{ap-definition-3}, the output of the hybrid assignment mechanism proposed in this paper is a stable solution in the Core, and any coalition $\mathcal{C}$ will not immediately form a large coalition. Therefore, for the convex game $G=(\mathcal{N},v)$, the result $(\mathcal{N},x)$ output by the hybrid mechanism is in the Core and is a stable reward reallocation.
\end{proof}

\subsection{Proof of Shapley Q-value Policy Gradient}
This paper introduces a historical interaction information-based Shapley Q-value policy gradient method, which aims to maximize the reward $ Q ^ { \Phi _ { i } } ( s , a ^ { i } )$ obtained by the current agent. This method effectively leverages historical interactions to improve both the efficiency and stability of the credit assignment process. The policy gradient can be defined as:
\begin{equation} \label{J_AP_Phi} 
    \! {\nabla_\theta}J_{i}\left(\theta \right)=\! \mathop \mathbb{E} \limits_{ \left({s}_t, {a}_t\right) \sim {\cal{D}}} \! \left[ \! \nabla_\theta  {\cal BC} \! \left( \! \log {\pi _\theta } \! \left({a}^i_t|s_t \right) \! \right)   Q_{\psi}^{\Phi_{i}}\left( s_t, {a}_t^i \right) \right]
\end{equation}
where $\cal{D}$ represents replay buffer, $\log {\pi _\theta } \! \left({a}^i_t|s_t \right)$ is the log probability of historical actions under the current policy.
Next, we derive the policy gradient algorithm as follows:
\begin{proof}
    In off-policy algorithms, the policy update only uses the action of the actor network $\pi_\theta$ sampling to calculate the gradient, and when the historical behaviors are used to calculate the gradient, the distribution of data will be inconsistent and policy entropy does not exist. The policy gradient for using historical data is as follows:
    \begin{equation} \label{eq_ade1}
        \begin{aligned}
        {\nabla_\theta} \! J_{i}\! \left(\!\theta \!\right) \!=\! {\mathop \mathbb{E}  \limits_{\scriptstyle \left( \! {s}_t,{a}_t \!\right) \sim \! {\cal{D}} \hfill}} \! \Bigl( \! {\nabla_\theta} \! \log \! {\pi_\theta} \! \left(\! {a}^i_t | s_t \! \right) \! Q_{\psi}^{\Phi_i} \! \left( \!{s}_t, {a}_t^i \!\right) \!  \!\Bigr)\!.\!
        \end{aligned}
    \end{equation}
    Next, we derive the policy gradient according to Eq. \ref{eq_ade1}.
    Assume that $X_{act}$ is a Gaussian random variable with mean $\mu$ and variance $\sigma^2$.
    $x_{act}$ satisfies $x_{act} \sim X_{act}$. Thus, $X_{act}$ and $x_{act}$ can be viewed as policy $\pi$ and action sampled from the policy. Then, the random variable $X_{act}$ satisfies:
    \begin{equation} \label{eq_gauss}
        \begin{aligned}
        f\left( x_{act} \right) = \frac{1}{\sqrt{2\pi} \sigma } {e^{-\frac{\left({x_{act}-\mu}\right)^2} {2{\sigma^2}}}}, \quad x_{act} \sim X_{act}.
        \end{aligned}
    \end{equation}
    Taking the logarithm of Eq. \ref{eq_gauss} gives the probability distribution:
    \begin{equation} \label{eq-loggauss}
        \begin{aligned}
        \log f\left( x_{act} \right) = - \frac{\left(x_{act} - \mu \right)^2}{2{\sigma ^2}} - \log \sigma - \log \left( {\sqrt {2\pi } } \right).
        \end{aligned}
    \end{equation}
    In Eq. \ref{eq-loggauss}, we can get:
    \begin{equation} \label{eq12}
        \begin{aligned}
        - \frac{\left(x_{act} - \mu \right)^2}{2{\sigma ^2}} = - \frac{1}{2}{\left( {\frac{{x_{act} - \mu }}{\sigma }} \right)^2}.
        \end{aligned}
    \end{equation}
    where ${\left(\frac{x_{act} -\mu}{\sigma}\right)^2}$ follows a Chi-square distribution with 1 degree of freedom. And when the standard deviation $\sigma$ is within a certain range, the formula $-\log \sigma - \log \left( {\sqrt {2\pi } } \right)$ in Eq. \ref{eq-loggauss} has only a slight influence on the log probability density function $\log f(x_{act})$.

    Although there are differences in the Gaussian policy distributions under different sampling states, the log probability distributions of the states sampled from the replay buffer in the off-policy algorithm and the actions sampled according to the current policy show a high degree of similarity. In order to maintain stability when optimizing the policy using historical action data, it is necessary to make the log policy distribution of historical actions approximately aligned with the log probability distribution of the currently sampled action.
    Therefore, we need to convert $\log \pi_\theta \left( a_t^i \mid s_t \right)$ into a specific distribution, while ensuring that this conversion preserves the relative order of the data. The Box-Cox transformation converts logarithmic probability data to an approximate normal distribution and does not change the relative size order. 
    We apply the Box-Cox transformation to the log probability $\log \! {\pi_\theta} \! \left(\! {a}^i_t | s_t \! \right)$ of the historical behavior $(s_t, a_t)$ and derive Eq. \ref{J_AP_Phi} that can use historical action data.
\end{proof}
\section{Implementation}
\begin{algorithm*}[t]
\caption{Historical Interaction-Enhanced Shapley Policy Gradient Algorithm}
\label{alg:HIS}
\textbf{Input}: sample times $M$, log adjustment factor $\beta$, batch size $B$, temperature $\alpha$, episodes $K$, steps per episode $T$, mini-epochs $e$, Polyak coefficient $\tau$, number of: agents $n$;\\
\textbf{Initialize}: replay buffer $\mathcal{D}$, the policy networks: $\left\{\theta^i\right\}_{i \in \mathcal{N}}$ and critic networks: $\psi_1$ and $\psi_2$, Set target parameters equal to main parameters $\psi_{\text {targ, } 1} \leftarrow \psi_1, \psi_{\text {targ, } 2} \leftarrow \psi_2$; \\
\begin{algorithmic} [1] 
\FOR{$k=1$ to $K-1$}
\STATE Observe state ${o}^i_t$ and select action ${a}^i_t \sim \pi_\theta^i(\cdot | {o}^i_t)$; \\
\STATE Execute ${a}^i_t$ in the environment;\\
\STATE Observe next state ${o}_{t+1}$, reward ${r}_t$; \\
\STATE Push transitions $\left\{\left({o}_t^i, {a}_t^i, {o}_{t+1}^i, {r}_t\right), \forall i \in \mathcal{N}, t \in T\right\}$ into $\mathcal{D}$; \\
\STATE Sample a random minibatch of $B$ transitions from $\mathcal{D}$; \\
\STATE Compute the critic targets: \\
\STATE
\[
            y_t=r+\gamma\left(\min _{i=1,2} Q_{\psi_{\text {targ }, i}}\left({s}_{t+1}, a_{t+1}\right)-\alpha \sum_{i=1}^{n}\log \pi^i_\theta\left({a}^i_{t+1} | {o}^i_{t+1}\right)\right), \quad a_{t+1} \sim  \pi_{\theta}\left(\cdot | {s}_{t+1}\right);
             \]
\STATE Update Q-functions by one step of gradient descent using \\
\STATE
    \[
        \psi_i=\arg \min _{\psi_i} \frac{1}{B} \sum_t\left(y_t-Q_{\psi_i}\left({s}_t, a_t\right)\right)^2 \quad \text { for } i=1,2;
        \]
\STATE Draw a permutation of agents $i_{1: n}$ at random; \\
    \FOR{agent $i_m=i_1$ to $i_n$}
    \STATE Sample $M$ ordered coalitions by $\mathcal{C} \sim Pr(\mathcal{C}|{\mathcal{N}}\backslash\{i\})$\\
        \FOR{each sampled coalition $\mathcal{C}_m$}
            \STATE Order each $a_t$ by $\mathcal{C}_m$ and mask the irrelevant agents' actions, storing it to $a^{m}_{t}$
        \ENDFOR
        \STATE Get Shapley Q-value by \\
        \[
            Q^{\Phi_i}(s_t, a_t^i) = \frac{1}{M} \sum_{m=1}^M \hat{\Phi}_{i} (s_t, a_t^m);
        \]
        \STATE Update agent $i_m$ by solving \\
        \[
            \begin{aligned}
            \theta^{i_m}_{\text{new}}=\arg \max _{\hat{\theta}^{i_m}} \frac{1}{B} \sum_t \Biggl(&\min _{i=1,2}Q_{\psi_i}\left({s}_t, a_{\boldsymbol{\theta}^{i_{1: m-1}}_{\text{new}}}^{i_{1: m-1}}\left(\mathbf{o}_t^{i_{1: m-1}}\right), {a}_{\hat{\theta}^{i_m}}^{i_m}\left({o}_t^{i_m}\right), a_{\boldsymbol{\theta}^{i_{m+1: n}}_{\text{old}}}^{i_{m+1: n}}\left(\mathbf{o}_t^{i_{m+1: n}}\right)\right)  \\
            & - \alpha\log\pi_{\hat{\theta}^{i_m}}^{i_m}\left({a}_{\hat{\theta}^{i_m}}^{i_m} | {o}_t^{i_m}\right)+{\cal BC} \! \Big( \! \log {\pi _\theta } \! \left({a}^{i_m}_t|s_t \right) \! \Big)   Q^{\Phi_{i_m}}\left( s_t, {a}_t^{i_m} \right)\Biggr)
            \end{aligned}
            \]
        \STATE where ${a}^i_\theta({o}^i_t)$ is a sample from $\pi^i_\theta(\cdot | {o}^i_t)$ which is differentiable wrt $\theta$ via the reparametrization trick; \\
        \STATE with $e$ mini-epochs of policy gradient ascent;
    \ENDFOR
    \STATE Update the target critic network smoothly
    \[
        \psi_{\text {targ }, i} \leftarrow \rho \psi_{\text {targ }, i}+(1-\rho) \psi_i \quad \text { for } i=1,2;
        \]
\ENDFOR
\STATE Discard $\psi$ . Deploy $\left\{\theta^i\right\}_{i \in \mathcal{N}}$ in execution;
\end{algorithmic}
\end{algorithm*}
The policy modeling of HIS follows the same structure as that of SAC, utilizing an unbounded Gaussian distribution with mean $u$ and standard deviation $\sigma$. Let the sampled action be denoted as $x_{act}$, and the Gaussian random variable corresponding to the policy $\pi$ is $X_{ACT}$, and $x_{act} \sim X_{ACT}$. Consequently, the distribution of the random variable $X_{ACT}$ can be expressed as:
\begin{equation} \label{eq_gauss2}
        \begin{aligned}
        f\left( x_{act} \right) = \frac{1}{\sqrt{2\pi} \sigma } {e^{-\frac{\left({x_{act}-\mu}\right)^2} {2{\sigma^2}}}}, \quad x_{act} \sim X_{act}.
        \end{aligned}
\end{equation}
Taking the logarithm of Eq. \ref{eq_gauss2} gives the probability distribution:
\begin{equation} \label{eq-loggauss2}
    \begin{aligned}
    \log f\left( x_{act} \right) = - \frac{\left(x_{act} - \mu \right)^2}{2{\sigma ^2}} - \log \sigma - \log \left( {\sqrt {2\pi } } \right).
    \end{aligned}
\end{equation}
Since the $\rm{tanh(\cdot )}$ will be applied to the actions sampled from the Gaussian distribution when outputting the actions, when HIS uses the historical logarithmic policy, it is necessary to first convert the current actions from the replay buffer back to the Gaussian distribution:
\begin{equation} \label{eq20}
    \begin{aligned}
    {a}^E_t = \frac{1}{2}\log \frac{1+{a}_t} {1-{a}_t}, \quad {a}_t \sim {\cal{D}}.
    \end{aligned}
\end{equation}
Next, we use the same method as SAC to calculate the likelihoods of historical actions:
\begin{equation} \label{eq19}
    \begin{aligned}
    &{\rm{log}} \pi_\theta\! \left(\!{a}_{t} |s_t\!\right)\! \\
    =& {\rm{log}}  f \left( \!{a}^E_t | s_t\! \right) \!\!- \!\! \!\sum\limits_{j = 1}^D \! {\rm{log}} \!\left(\!1\!-\!{\rm{tanh^2}} \!\left({a}^E_{t,j} \right) \right)\\
    =&{\rm{log}}  f \left( \!{a}^E_t | s_t\! \right) \!\!- \!\! \!\sum\limits_{j = 1}^D \!{\rm{log}} \!\left( {\rm{sech}}^2 \!\left({a}^E_{t,j} \right) \right)\\
    =&{\rm{log}}  f \left( \!{a}^E_t | s_t\! \right) \!\!- \!\! \!\sum\limits_{j = 1}^D \! 2{\rm{log}} \!\left( \frac{2}{e^{{a}^E_{t,j}} + e^{-{a}^E_{t,j}}} \right) \\
    =&{\rm{log}}  f \left( \!{a}^E_t | {s}_t\! \right) \!\!- \!\! \!\sum\limits_{j = 1}^D \!
    2\left({\rm{log}}2 - {a}^E_{t,j} - {\rm{Softplus}} \left(-2{a}^E_{t,j}\right) \right).
    \end{aligned} 
\end{equation}
where $\rm{Softplus(\cdot)}$ refers to an activation function often used in neural networks;
$a_t^E$ is a D-dimensional action. $a_{t,j}^E$ is the j-th element of $a_t^E$. In order to prevent $\log f(a_t^E|s_t)$ from causing training instability, we limit the size of $\log f(a_t^E|s_t)$ during this process. We define the minimum threshold of $\log f(a_t^E|s_t)$ as:
\begin{equation}
    T_{limit}=max(\log f(a_t^E|s_t))-\beta.
\end{equation}
where $\beta$ is an adjustment factor introduced by the algorithm to adjust the threshold to adapt to changes in the environment. When $\log f(a_t^E|s_t)$ is less than $T_{limit}$, we impose the following restriction:
\begin{equation} \label{5_add_1}
    \begin{aligned}
    \log f \left( {a}^E_t| s_t \right)  \! = \! T_{limit}  \! +  \! e^{ \log f \left( {a}^E_t| s_t \right) - \! P_{limit} }-1.
    \end{aligned}
\end{equation}

After computing the logarithmic policy of historical behaviors, the data is transformed using the Box-Cox transformation to enhance its normality and symmetry, thereby promoting the stability of policy training.
It is worth noting that in order to make the transformed data more consistent with the original data distribution, we will add $x_{min}$ in the BOX-COX transformation to the transformed result.
The pseudocode of HIS is shown in Algorithm \ref{alg:HIS}.

\section{Experimental Details}
\subsection{Additional Experiments}
\subsubsection{Bi-DexHands}
Bi-DexHands \cite{chen2022towards} consists of a series of two-handed robotic arm manipulation tasks that demand high control accuracy and collaborative performance from the algorithm. To assess the effectiveness of the HIS algorithm, we compared it with several baseline algorithms, including HASAC, HATRPO, HAPPO, MAPPO, and PPO \cite{schulman2017proximal}, across nine different tasks.

As shown in Figure \ref{fig:DexHands}, the HIS algorithm consistently outperforms the baselines in terms of control accuracy and overall performance. Notably, it excels in complex, highly demanding tasks, such as the Catch Abreast, Two Catch Underarm, and Lift Underarm tasks. These tasks are particularly challenging due to their intricate coordination requirements and precision control, yet HIS maintains competitive performance in these environments.

The superior performance of HIS can be attributed to its hybrid credit assignment mechanism, which differentiates agent contributions through historical interaction information. Unlike traditional shared reward schemes, where agents receive identical feedback, the HIS algorithm enables each agent to learn behaviors that maximize its individual interests while contributing effectively to the global objective. This differentiation of contributions allows HIS to perform more effectively in tasks requiring fine-grained control and coordination, providing a distinct advantage over the baseline algorithms.

\begin{table}[!htb]
\centering
\newcolumntype{C}{>{\centering\arraybackslash}X}
\begin{tabularx}{0.5\textwidth}{CC} 
\hline
Hyperparameters & Value \\ \hline
proper time limits       & True           \\
warmup steps             & $10^4$         \\
activation               & ReLU           \\
final activation         & Tanh           \\
buffer size              & $10^6$         \\
polyak                   & 0.005          \\
hidden sizes             & [256, 256]     \\
update per train         & 1              \\
train interval           & 50             \\
target entropy           & $-\text{dim}(\mathcal{A}^i)$ \\
policy noise             & 0.2            \\
noise clip               & 0.5            \\
policy update frequency  & 2              \\
linear lr decay          & False          \\ \hline
\end{tabularx}
\caption{Common hyperparameters used for off-policy algorithms HIS across all environments.}
\label{tab:tab1}
\end{table}


\begin{table}[!htb]
\centering
\newcolumntype{C}{>{\centering\arraybackslash}X}
\begin{tabularx}{0.5\textwidth}{CC} 
\hline
Hyperparameters & Value \\ \hline
rollout threads & 20           \\
gamma  & 0.95        \\
batch size & 1000           \\
actor lr        & 5e - 4           \\
critic lr & 5e - 4         \\
n step &   20          \\
use huber loss	& False     \\
use valuenorm &	False              \\
sample times & 2           \\
log adjustment factor  & 10        \\
exploration noise  & 0.1            \\ \hline
\end{tabularx}
\caption{Common hyperparameters used for HIS in the Bi-DexHands domain.}
\label{tab:bi1}
\end{table}

\begin{table}[!htb]
\centering
\begin{tabular}{c|cccc}
\hline
scenarios & auto alpha & alpha  & alpha lr        \\ \hline
Catch Abreast       & True      & / & 3e - 4 \\
Two Catch Underarm    & False      & 5e - 5 & / \\ 
Hand Lift Underarm       & True      & / & 3e - 4 \\
Hand Over       & True      & / & 3e - 4 \\ 
Catch Over2Underarm       & True      & / & 3e - 4 \\ 
Hand Pen       & True      & / & 3e - 4 \\  
Door Close Inward       & True      & / & 3e - 4 \\ 
Door Open Inward       & True      & / & 3e - 4 \\ 
Door Open Outward       & True      & / & 3e - 4 \\ \hline
\end{tabular}
\caption{Different hyperparameters used for HIS in the Bi-DexHands domain.}
\label{tab:bi2}
\end{table}

\subsubsection{Multi-Agent MuJoCo (MAMuJoCo)}
MAMuJoCo is a widely recognized benchmark environment for heterogeneous agent collaboration, encompassing a diverse set of robot motion tasks \cite{de2020deep}. In this setup, each joint of the robot is treated as an independent agent, and the algorithm must manage the coordination and control of multiple interdependent joints to enable the robot to complete complex motion tasks. To evaluate the performance of the HIS algorithm, we conducted a comparison against baseline algorithms across nine tasks distributed across four distinct scenarios.

To highlight the advantages of the hybrid credit assignment mechanism used by HIS, we specifically compared it to the local reward mechanism in strongly coupled environments. We incorporated the FACMAC \cite{peng2021facmac}, which employs an implicit decomposition scheme, as a baseline. As the degree of coupling between the robot joints increases, estimating the contribution of global rewards becomes more challenging. 
However, the FACMAC implementation does not cover all the tasks we use, and we compare the results of six of them with the same tasks as ours.
As can be seen in Figure 5, in tasks such as HalfCheetah-v2-6x1 and manyagent\_swimmer-10x2, FACMAC, due to its reliance on local decomposition schemes, has difficulty accurately estimating the contributions of individual agents. As a result, it is unable to learn effective policies for these complex and interdependent tasks.

On the other hand, while algorithms such as HASAC, HATRPO, HAPPO, and MAPPO, which rely on the shared reward scheme, perform reasonably well, they still face significant limitations in credit assignment. These methods tend to blur the individual contributions of agents to the global reward, which results in less efficient learning and suboptimal policy development. The HIS algorithm, which combines both shared and local reward schemes, provides a more nuanced approach to credit assignment.

Experimental results show that HIS exhibits strong robustness and stability in both simple and complex scenarios. It consistently achieves state-of-the-art performance, outperforming baseline algorithms in all nine tasks. The hybrid credit assignment mechanism allows HIS to effectively handle the complexities of strongly coupled tasks by more accurately differentiating individual contributions, leading to improved learning efficiency and more effective strategy development.
\begin{table}[!htb]
\centering
\newcolumntype{C}{>{\centering\arraybackslash}X}
\begin{tabularx}{0.5\textwidth}{CC} 
\hline
Hyperparameters & Value \\ \hline
rollout threads & 20           \\
batch\_size     & 1000        \\
critic lr       & 5e - 4           \\
gamma           & 0.99          \\
actor lr        & 5e - 4        \\
n\_step         & 20 \\
auto alpha      & True \\
alpha lr        & 3e - 4 \\ 
use huber loss	& False \\
use valuenorm &	False         \\
sample times & 3          \\
use valuenorm &	False \\
log adjustment factor  & 10        \\
 \hline
\end{tabularx}
\caption{Hyperparameters used for HIS in the MPE domain.}
\label{tab:mpe}
\end{table}

\subsection{Hyper-Parameter Settings For Experiments}
\subsubsection{Common Hyper-Parameter}
The Boolean variable auto\_alpha is used to indicate whether to perform automatic tuning. When it is set to True, the target entropy and alpha\_lr will take effect. When it is set to False, the temperature coefficient $\alpha$ is set to an arbitrary value.
The HIS algorithm is developed based on the HARL \cite{zhong2024heterogeneous} framework. The implementation details of HASAC in our baseline are as described in the HASAC literature \cite{liu2023maximum}, while HAPPO, HATRPO, and MAPPO are as described in the HARL literature. Next, we provide the hyperparameters commonly used by HIS in all environments in Table \ref{tab:tab1}. 

\begin{table*}[t]
\centering
\begin{tabular}{c|cccccc} \hline
scenarios               & actor lr & auto alpha & alpha            & alpha lr          & batch size & n\_step\\ \hline
Ant 2x4                 & 3e - 4   & False       & 0.2              & /    & 2200      & 5\\
Ant 4x2                 & 3e - 4   & False       & 0.2              & /   & 1000  & 5     \\
Ant 8x1                 & 3e - 4   & False       & 0.2              & /   & 2200  & 5     \\Walker 2x3              & 5e - 4   & False       & 0.2              & /  & 1000       & 20\\
Walker 6x1              & 5e - 4   & False       & 0.2              & /   & 1000  & 20     \\
HalfCheetah  2x3        & 1e - 3   & True        & / & 3e - 4              & 1000       & 10\\
HalfCheetah 3x2         & 1e - 3   & True        & / & 3e - 4             & 1000       & 10\\
HalfCheetah 6x1         & 1e - 3   & True        & / & 3e - 4            & 1000 & 10      \\
manyagent\_swimmer 10x2 & 1e - 3   & True        & / & 3e - 4             & 1000       & 10\\ \hline
\end{tabular}
\caption{Different hyperparameters used for HIS in the MAMuJoCo domain.}
\label{tab2}
\end{table*}

\subsubsection{Bi-DexHands}
In this setting, we used the implementation and tuned hyperparameters of the baseline algorithms PPO, MAPPO, HATRPO, HAPPO described in the HARL paper \cite{zhong2024heterogeneous}. The implementation and hyperparameters of HASAC refer to its original paper \cite{liu2023maximum}. Table \ref{tab:bi1} and Table \ref{tab:bi2} give the hyperparameters of HIS in different experiments in Bi-DexHands.

\subsubsection{Multi-Agent Particle Environment (MPE)}
Table \ref{tab:mpe} shows the hyperparameters used by HIS in the MPE task. For the HASAC algorithm, we adopted the implementation and tuning hyperparameters in HASAC \cite{liu2023maximum}. The implementation details and hyperparameters of other baseline algorithms refer to the paper \cite{zhong2024heterogeneous}.

\subsubsection{Multi-Agent MuJoCo (MAMuJoCo)}
The experimental details and tuned hyperparameters of FACMAC are consistent with the original paper \cite{peng2021facmac}. For the HASAC algorithm, we adopted the implementation details and tuned hyperparameters in the HASAC paper \cite{liu2023maximum}. The implementation details and tunable hyperparameters of other baseline algorithms refer to the description in HARL \cite{zhong2024heterogeneous}. Next, the hyperparameters of HIS for multiple tasks in MAMujoco are shown in Table \ref{tab2} and Table \ref{tab3}.

\begin{table}[!htb]
\centering
\newcolumntype{C}{>{\centering\arraybackslash}X}
\begin{tabularx}{0.5\textwidth}{CC} 
\hline
Hyperparameters & Value \\ \hline
rollout threads & 10           \\
train interval  & 50        \\
critic lr       & 1e - 3           \\
gamma           & 0.99          \\
use huber loss	& False        \\
sample times & 1          \\
log adjustment factor  & 10        \\
use valuenorm &	False         \\ \hline
\end{tabularx}
\caption{Common hyperparameters used for HIS in the MAMuJoCo domain.}
\label{tab3}
\end{table}

\section{Hyperparameter Robustness Experiments}
To investigate the impact of key hyperparameters in the HIS algorithm on its stability, we conducted statistical analysis on the sampling times $M$ and the log adjustment factor $\beta$ of the main hyperparameters in HIS in three benchmark tasks. We explored the impact of different values of these hyperparameters on the algorithm effect.
Figure \ref{fig:hyper-sample} and Figure \ref{fig:hyper-log} reflects the statistical analysis results for different values of these two parameters, where the x-axis represents the different values of the hyperparameters and the y-axis represents the average return. In addition, we use the average return of the last 50\% of the total training steps to obtain more reliable statistical results that reflect the steady-state performance of the algorithm.
\begin{figure}[htbp]
    \centering
    \subfloat{\includegraphics[width=0.33\linewidth]{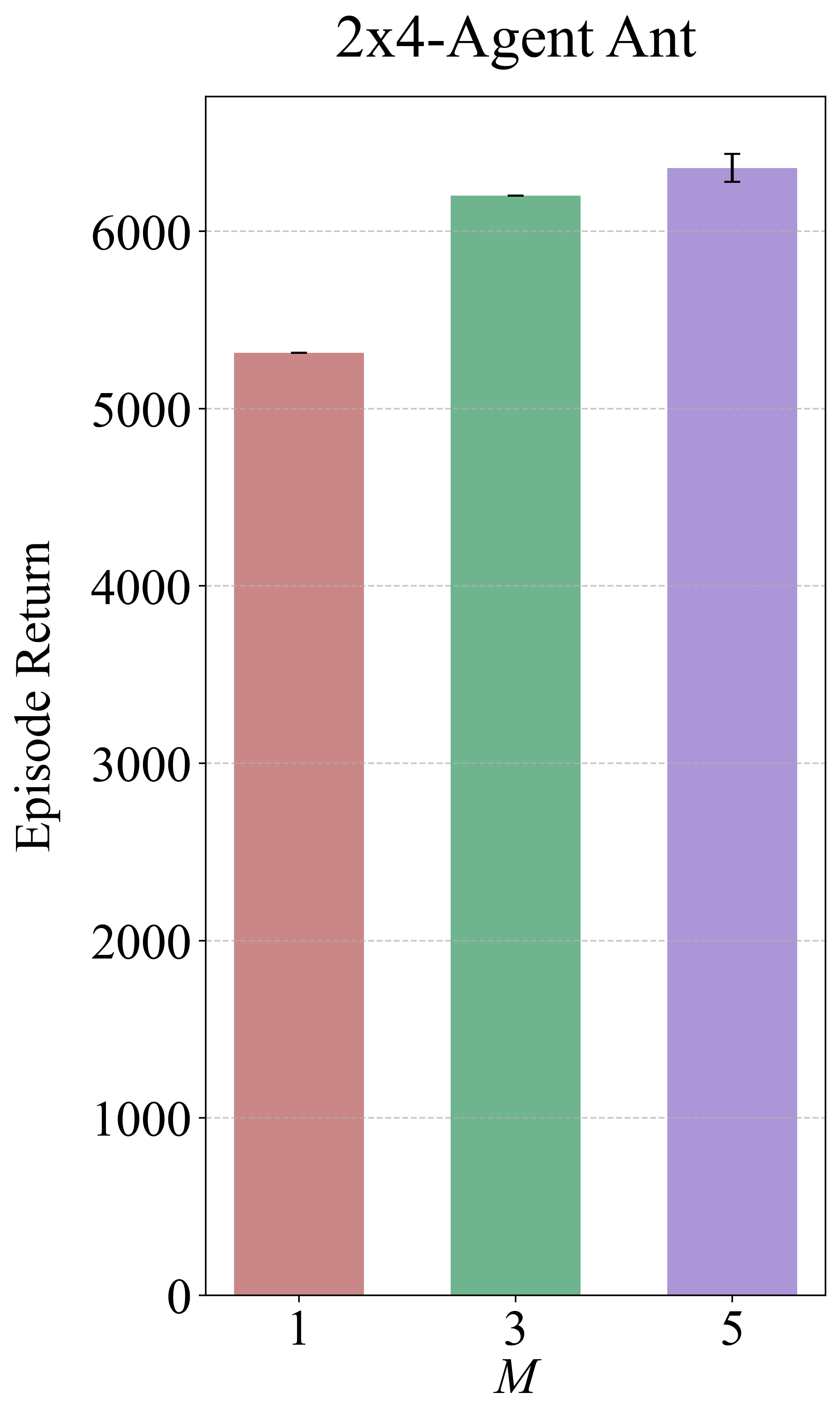}}
    \subfloat{\includegraphics[width=0.33\linewidth]{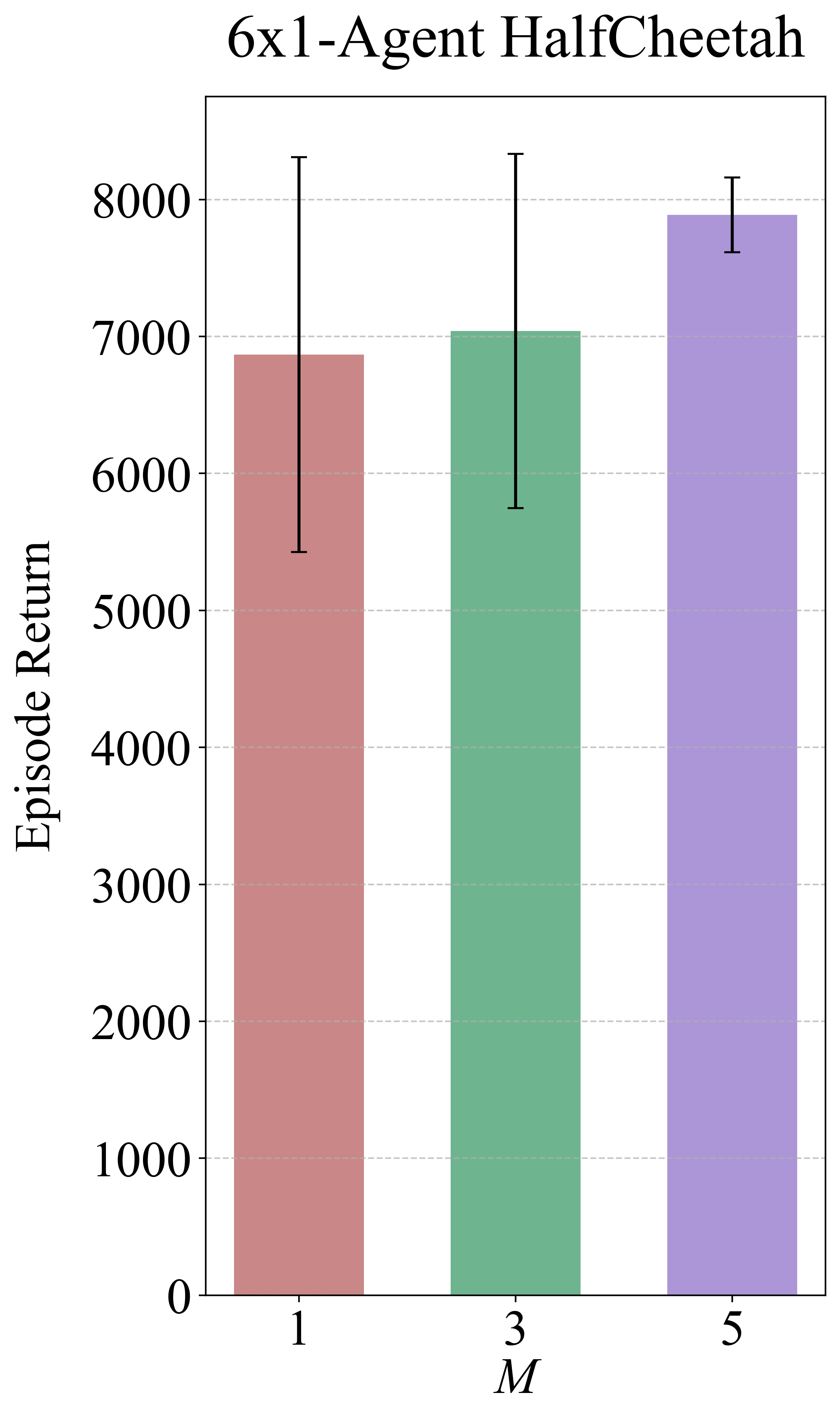}}
    \subfloat{\includegraphics[width=0.33\linewidth]{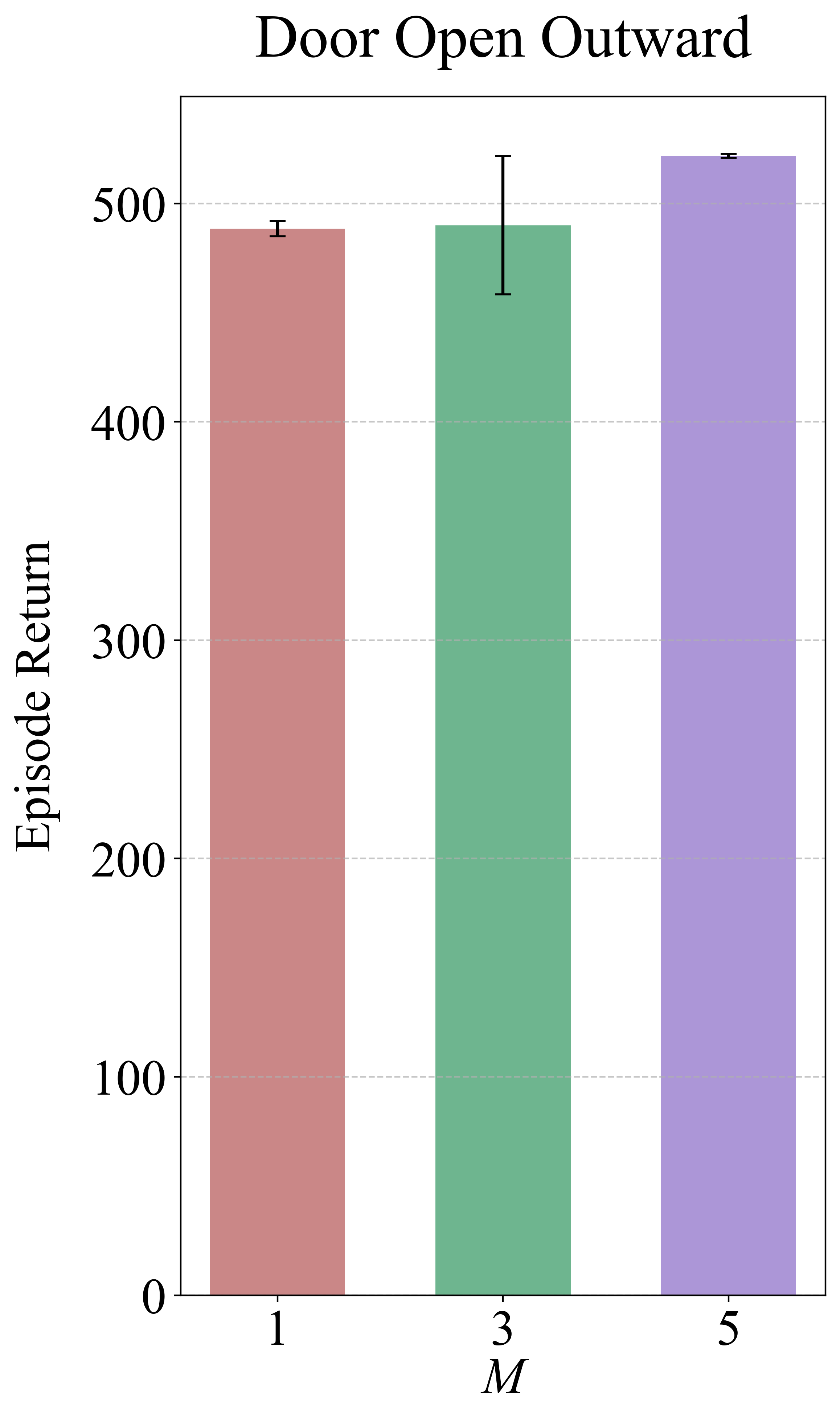}}\\
    \caption{Statistical analysis of $M$ in various tasks.}
    \label{fig:hyper-sample}
\end{figure}

As observed in Figure \ref{fig:hyper-sample}, as the sampling times $M$ increases, the estimation of marginal contribution becomes more accurate, leading to a corresponding improvement in the algorithm's performance.
However, a higher value of $M$ also results in increased computational cost. 
Despite this trade-off, the algorithm consistently maintains full stability across different values of $M$, indicating that the HIS algorithm is robust to variations in this hyperparameter.

\begin{figure}[htbp]
    \centering
    \subfloat{\includegraphics[width=0.33\linewidth]{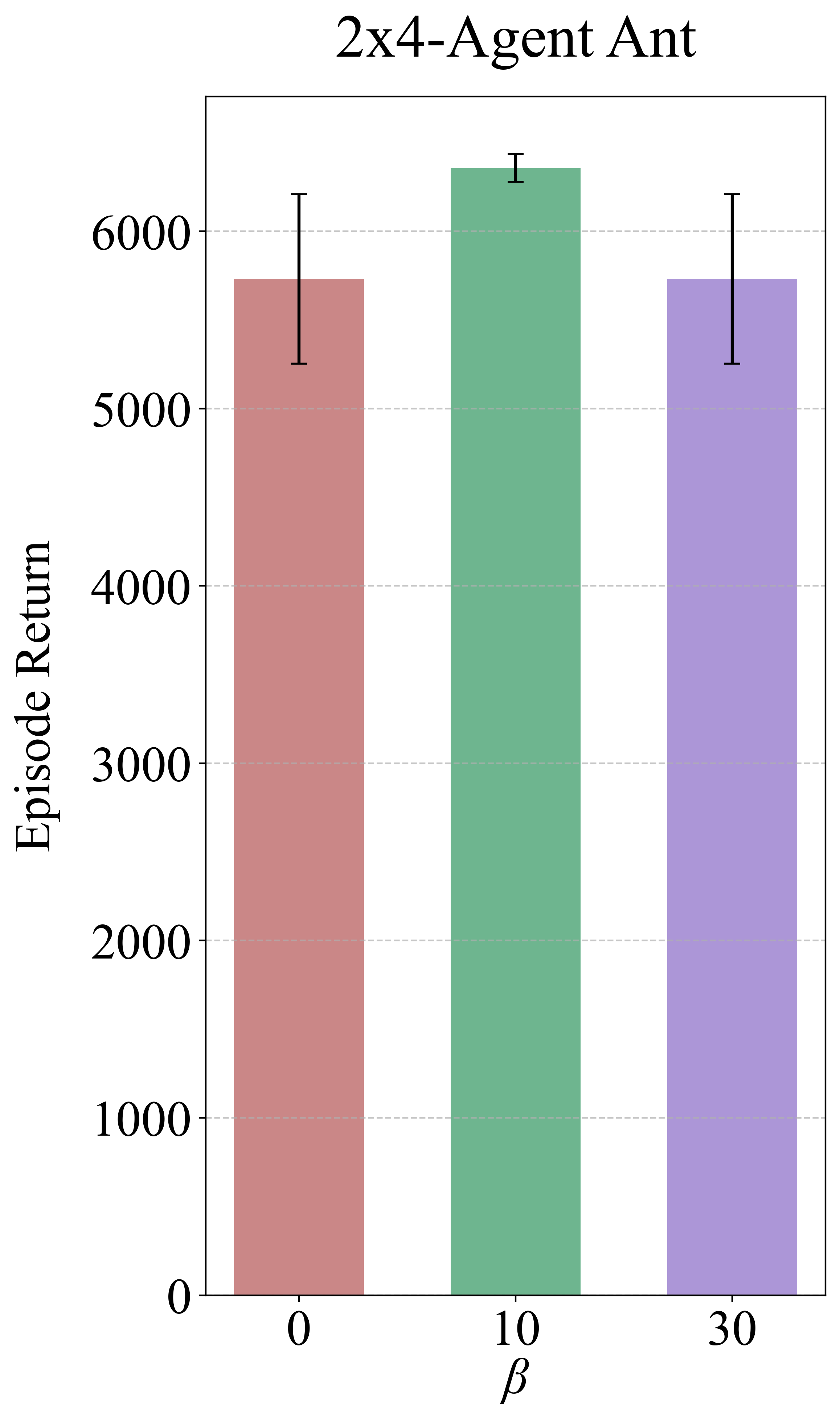}}
    \subfloat{\includegraphics[width=0.33\linewidth]{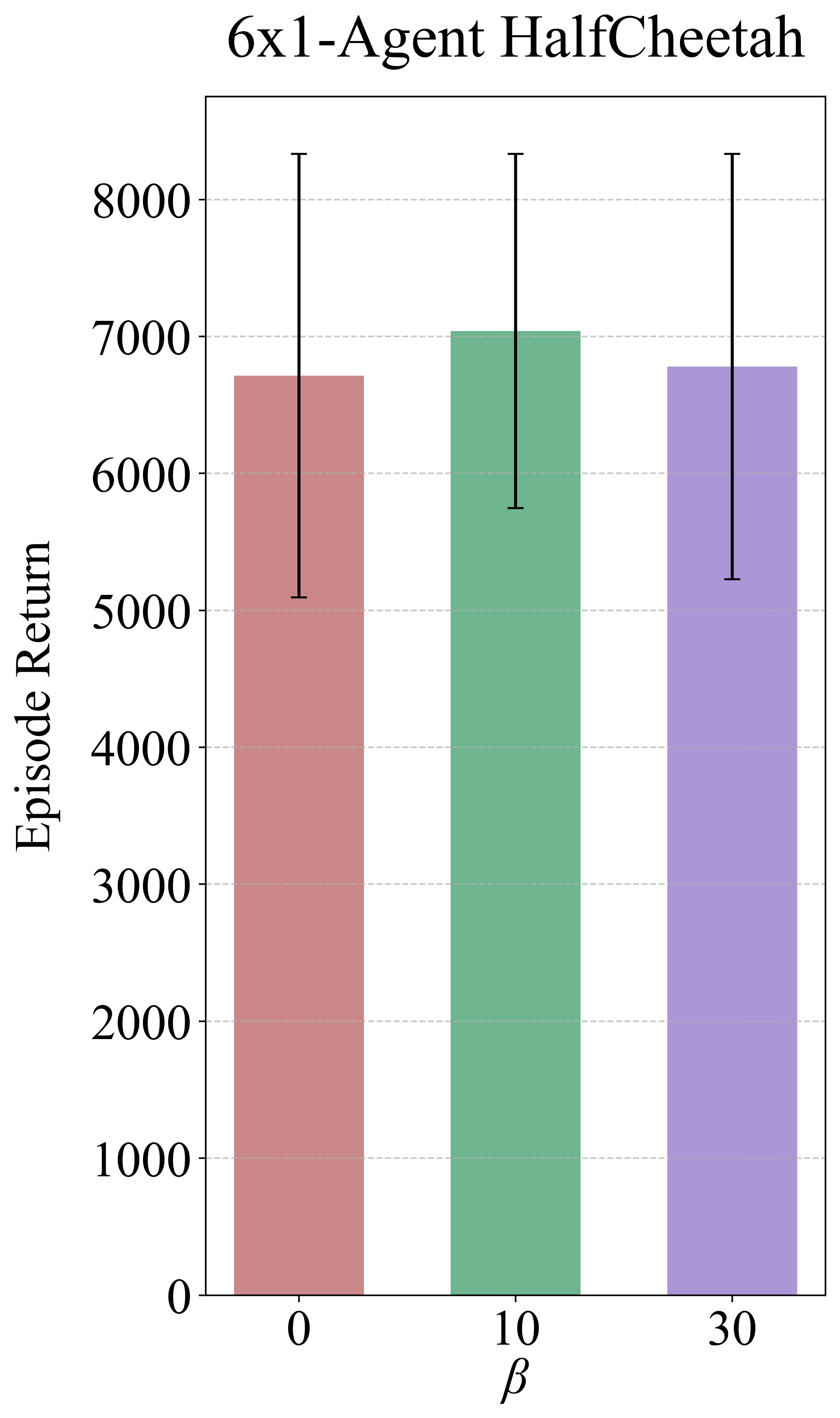}}
    \subfloat{\includegraphics[width=0.33\linewidth]{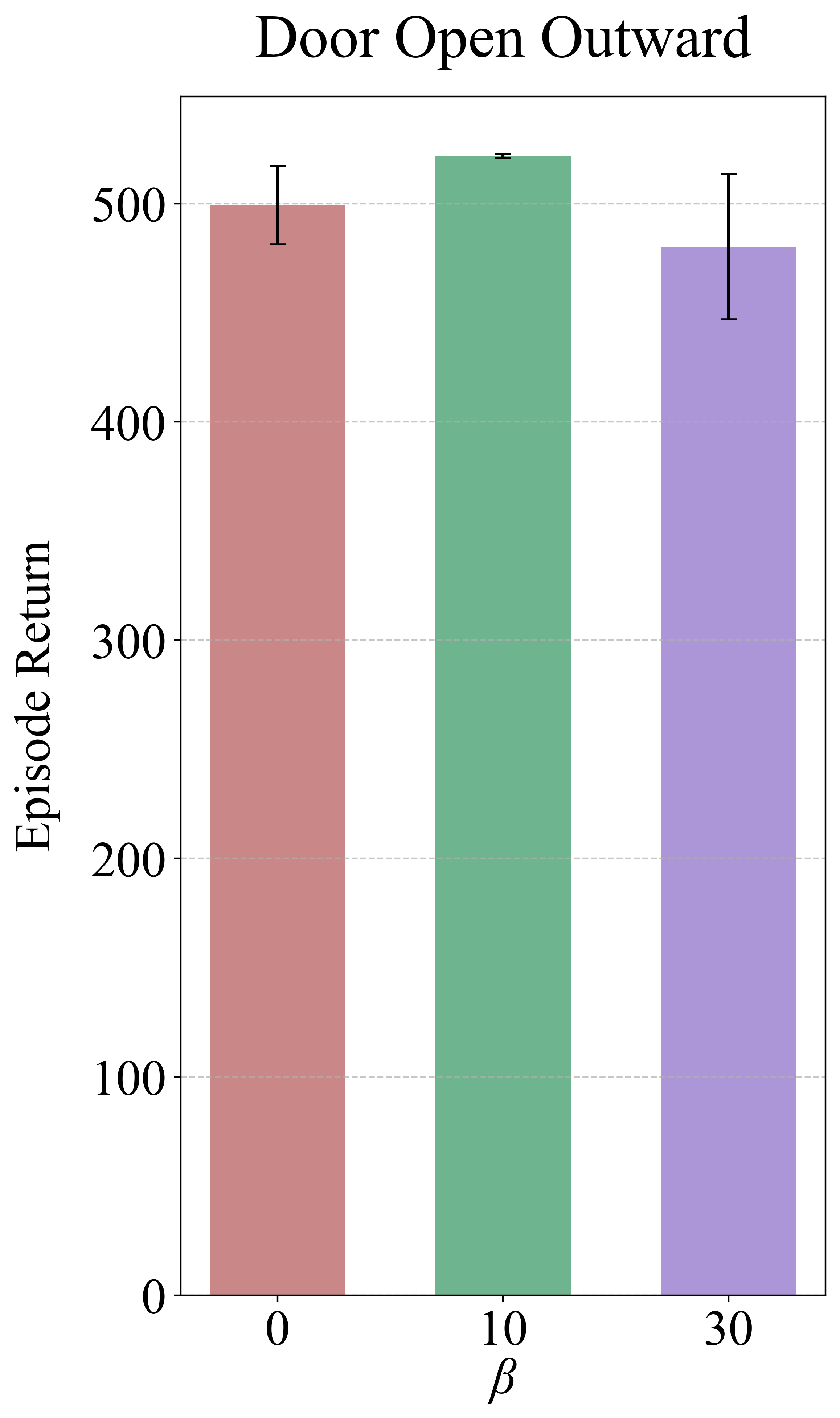}}\\
    \caption{Statistical analysis of $\beta$ in various tasks.}
    \label{fig:hyper-log}
\end{figure}
As shown in Figure \ref{fig:hyper-log}, while the performance of HIS exhibits some fluctuation with changes in $\beta$, when $\beta$ is within a specific range, the performance of the algorithm will not fluctuate greatly with the change of hyperparameters.

These results show that HIS is robust to changes in its key hyperparameters, and it maintains consistent performance under different hyperparameter settings.
This suggests that even if the precise configuration of parameters deviates from the optimal value, the algorithm's underlying mechanisms—particularly the hybrid credit assignment and historical interaction enhancement—can provide reliable guidance for policy optimization.
\section{Additional Ablation Experiments of HIS}
HIS uses the BOX-COX transformation to align data distribution and uses historical actions to estimate Shapley Q-value. This section uses ablation experiments to answer the following two questions:
\begin{figure*}[t] 
    \centering
    \subfloat[{2x4-Agent Ant-v2}]{%
        \includegraphics[width=0.33\linewidth]{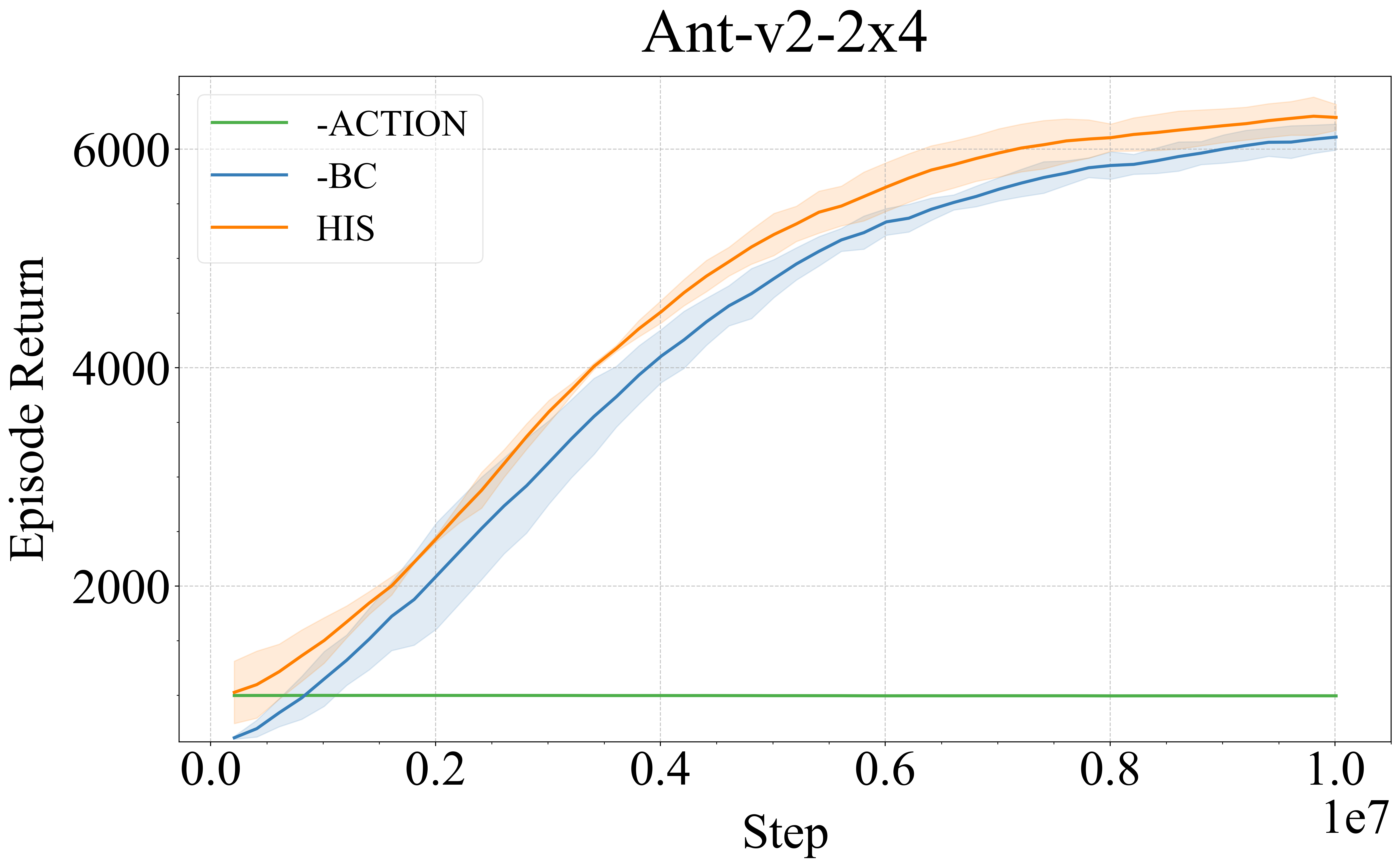}%
    }%
    \hfill%
    \subfloat[{8x1-Agent Ant-v2}]{%
        \includegraphics[width=0.33\linewidth]{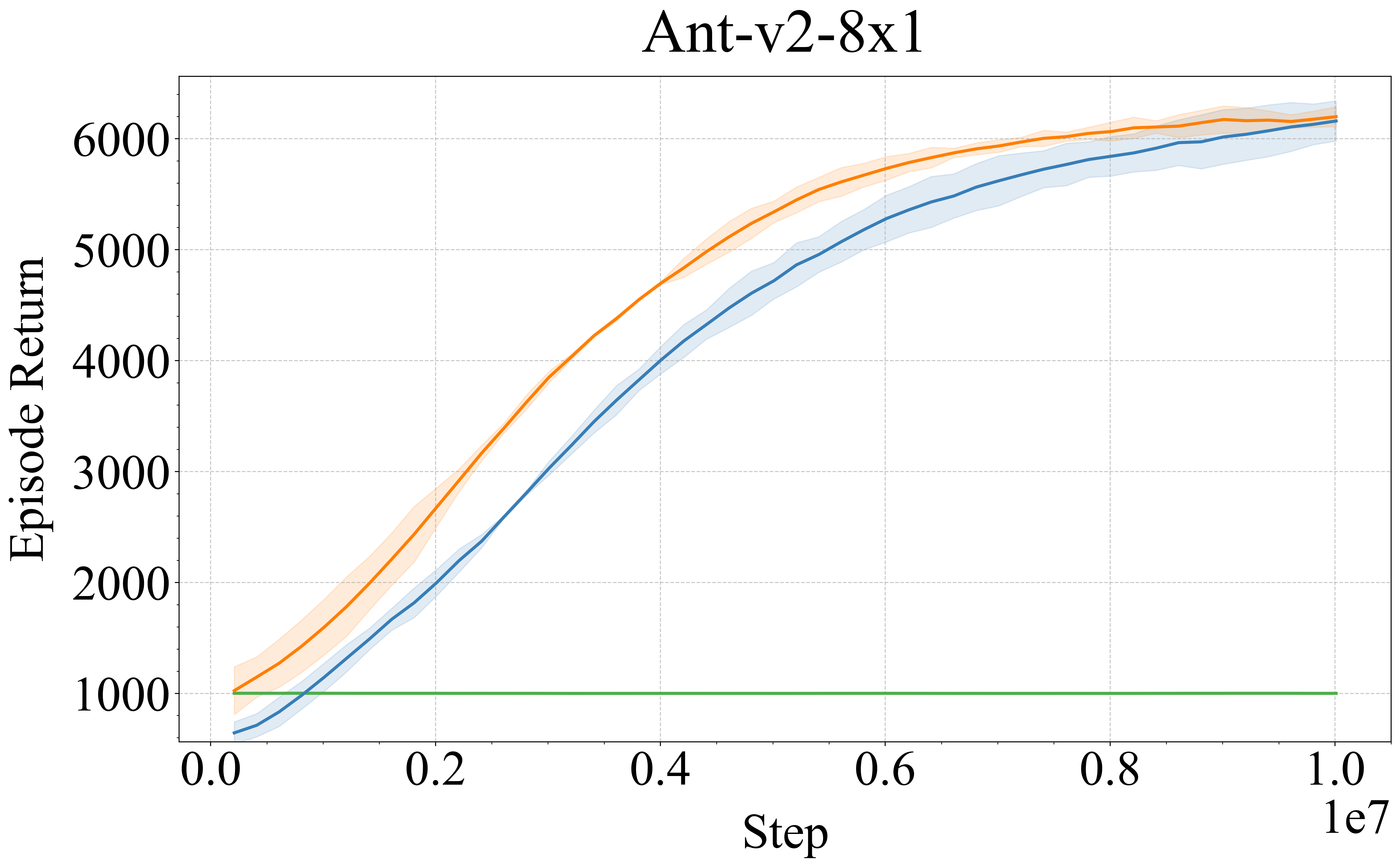}%
    }%
    \hfill%
    \subfloat[{Door Open Outward}]{%
        \includegraphics[width=0.33\linewidth]{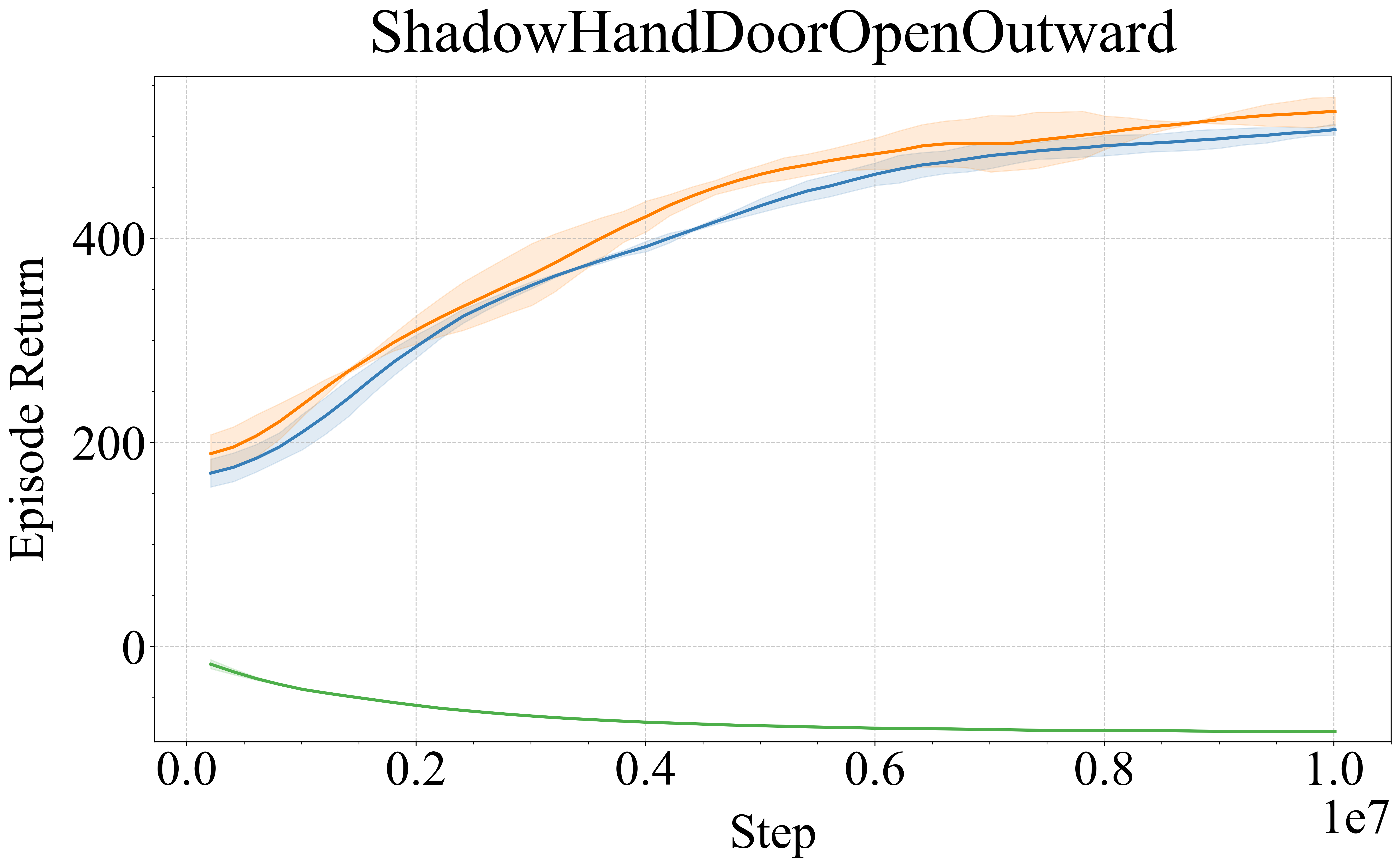}%
    }
    \caption{Ablation study of different HIS components on Bi-DexHands and MAMuJoCo tasks.}
    \label{fig:add-ablation}
\end{figure*}

\begin{enumerate}[(1)]
    \item ``Can the data distribution alignment function of BOX-COX effectively improve the overall performance of the algorithm?"
    \item ``What impact will it have on the algorithm performance if the actions sampled by the current policy are used to estimate Shapley Q-value?"
\end{enumerate}


To answer these two questions, we compare HIS with the following two HIS variants on the Bi-DexHands and MAMuJoCo benchmark environments:

\textbf{-BC}: When calculating the policy gradient, the Box-Cox transformation is not used to align the data distribution. Specifically, the policy gradient form of this variant is as follows:
\begin{equation}
    \label{AP-J-BC}
    \begin{split}
    &{\nabla_\theta}J_{\pi^{i_m}}(\theta^{i_m}) \\
    &= \mathbb{E}_{s_t \sim \mathcal{D}} \left[ \mathbb{E}_{{a}_t^{i_{1:m-1}} \sim \pi_{\theta_{new}^{i_{1:m-1}}}^{i_{1:m-1}}, {a}_t^{i_m} \sim \pi_{\theta^{i_m}}^{i_m}} \bigg[ \right. \\
    &\quad \nabla_\theta \Big(  Q_{\pi_{\text{old}}:\psi}^{i_{1:m}} \left( s_t, {a}_t^{i_{1:m-1}}, {a}_t^{i_m} \right)  - \alpha \log \pi_{\theta^{i_m}}^{i_m} \left( {a}_t^{i_m} \mid s_t \right) \Big) \bigg] \Bigg]  \\
    &\quad + \mathop \mathbb{E} \limits_{ \left({s}_t, {a}_t\right) \sim {\cal{D}}} \left[ \nabla_\theta  \Big( \log {\pi _\theta } \left({a}^{i_m}_t|s_t \right) \Big)   Q_{\psi}^{\Phi_{i_m}}\left( s_t, {a}_t^{i_m} \right) \right].
    \end{split}
\end{equation}
\textbf{-ACTION}: When calculating the policy gradient, the data in the historical action will no longer be used to estimate the Shapley Q-value. 
Instead, the actions sampled by the policy in the sequential update scheme will be used to estimate the Shapley Q-value. Specifically, the policy gradient form of this variant is as follows:
\begin{equation}
    \label{AP-J-Action}
    \begin{split}
    &{\nabla_\theta}J_{\pi^{i_m}}(\theta^{i_m}) \\
    &= \mathbb{E}_{s_t \sim \mathcal{D}} \left[ \mathbb{E}_{{a}_t^{i_{1:m-1}} \sim \pi_{\theta_{new}^{i_{1:m-1}}}^{i_{1:m-1}}, {a}_t^{i_m} \sim \pi_{\theta^{i_m}}^{i_m}} \bigg[ \right. \\
    &\quad \nabla_\theta \Big(  Q_{\pi_{\text{old}}:\psi}^{i_{1:m}} \left( s_t, {a}_t^{i_{1:m-1}}, {a}_t^{i_m} \right)  - \alpha \log \pi_{\theta^{i_m}}^{i_m} \left( {a}_t^{i_m} \mid s_t \right) \\
    &\quad + \Big( \log {\pi _\theta } \left({a}^{i_m}_t|s_t \right) \Big) Q_{\psi}^{\Phi_{i_m}}\left( s_t, {a}_t^{i_m} \right) \Big) \bigg] \Bigg]  
    \end{split}
\end{equation}

Figure \ref{fig:add-ablation} shows that -BC's average return across all three tasks is significantly lower than HIS's. For example, in Ant-v2-8x1, -BC's convergence speed and training stability are significantly worse than those of HIS. 
This result highlights the importance of applying the Box-Cox transformation to historical data. The Box-Cox transformation strengthens the normality and symmetry of the data without changing the relative order of magnitude, improving the stability and effectiveness of algorithm training.

Furthermore, experimental results show that -ACTION fails to learn effective policies in these three tasks, resulting in very poor performance. This is likely due to the algorithm's instability caused by using actions sampled by the current policy to estimate Shapley Q-value and perform policy learning. 
The agent is unable to estimate a Shapley Q-value that accurately reflects its contribution based on the currently retrieved actions. Consequently, the algorithm is unable to accurately assign credit to multiple agents, resulting in ineffective learning.



\end{document}